\documentclass[11pt,twoside,letter]{article}
\usepackage{amsthm, amsmath, amssymb, amsfonts}
\usepackage{mathtools}
\usepackage{booktabs}

\usepackage{bm}
\usepackage{bbm}
\usepackage{mathrsfs}
\usepackage{cancel}

\usepackage{graphicx}

\usepackage{caption}
\usepackage[font=sl,labelfont=bf]{caption}
\usepackage{subcaption}

\usepackage{epstopdf}
\usepackage[shortlabels]{enumitem}

\usepackage[usenames,dvipsnames]{color}

\definecolor{labelkey}{rgb}{0,0,1}

\usepackage[colorlinks=true, pdfstartview=FitV, linkcolor=blue, citecolor=blue, urlcolor=blue]{hyperref}
\usepackage[usenames]{color}

\usepackage{url}
\makeatletter\def\url@leostyle{%
 \@ifundefined{selectfont}{\def\UrlFont{\sf}}{\def\UrlFont{\scriptsize\ttfamily}}} \makeatother

\usepackage[margin=1.0in, letterpaper]{geometry}

\newtheorem{theorem}{Theorem}
\newtheorem{proposition}[theorem]{Proposition}

\theoremstyle{definition}
\newtheorem{definition}[theorem]{Definition}

\theoremstyle{remark}
\newtheorem{remark}[theorem]{Remark}

\numberwithin{equation}{section}
\numberwithin{theorem}{section}

\definecolor{Red}{rgb}{0.9,0,0.0}
\definecolor{Blue}{rgb}{0,0.0,1.0}

\def\cA{\mathcal{A}}
\def\cB{\mathcal{B}}
\def\cC{\mathcal{C}}

\def\cI{\mathcal{I}}

\def\cL{\mathcal{L}}

\def\cS{\mathcal{S}}

\def\bP{\mathbb{P}}

\def\bR{\mathbb{R}}

\def\sB{\mathscr{B}}

\def\sF{\mathscr{F}}

\def\sS{\mathscr{S}}

\def\bfB{\mathbf{B}}
\def\bfC{\mathbf{C}}

\def\bfS{\mathbf{S}}

\def\bfW{\mathbf{W}}

\def\bfpi{\boldsymbol{\pi}}

\newcommand{\1}{\mathbbm{1}}            
\newcommand{\set}[1]{\{#1\}}            

\DeclareMathOperator*{\argmin}{arg\,min} 
\DeclareMathOperator*{\argmax}{arg\,max} 
\DeclareMathOperator{\Var}{\mathrm{Var}}          

\title{Pro-rata mechanisms in groundwater markets}

\author{ 
	Igor Cialenco\,\thanks{Department of Applied Mathematics, Illinois Institute of Technology
		\newline \hspace*{1.45em}  10 W 32nd Str, Building RE, Room 220, Chicago, IL 60616, USA
		\newline \hspace*{1.45em}  Email: \url{cialenco@iit.edu},  URL: \url{http://cialenco.com}
        \newline \hspace*{1.45em} ORCID: \url{https://orcid.org/0000-0002-1825-0097}
		\vspace{0.5em}} 
	\and   
	 Michael Ludkovski\,\thanks{Department of Statistics and Applied Probability, University of California Santa Barbara
		\newline \hspace*{1.45em}  South Hall, Santa Barbara, CA 93106-3110, USA
		\newline \hspace*{1.45em} Email: \url{ludkovski@pstat.ucsb.edu}, URL: \url{http://ludkovski.pstat.ucsb.edu/}, 
        \newline \hspace*{1.45em} ORCID: \url{https://orcid.org/0000-0001-7887-3870}
		\vspace{0.5em}}
\and
 Gael Dimitri Tekam Fongouo\,\thanks{Department of Applied Mathematics, Illinois Institute of Technology
		\newline \hspace*{1.45em}  10 W 32nd Str, Building RE, Room 220, Chicago, IL 60616, USA
		\newline \hspace*{1.45em}  Email: \url{gtekamfongouo@hawk.illinoistech.edu}
		\vspace{0.5em}} 
}

\date{ {\small  First Circulated and This Version: August 1, 2026}}

\begin{document}

	\maketitle

	\vspace{-2em}

	{\footnotesize
		\begin{tabular}{l@{} p{350pt}}
			\hline \\[-.2em]
			\textsc{Abstract}: \ & 
We introduce a pro-rata rationing mechanism for resolving supply-demand imbalances in groundwater markets, extending the price-formation model of Cialenco and Ludkovski (2025).  We show that under the pro-rata distribution, every price is a Nash equilibrium, thereby pro-rata approach provides a rationing device whenever supply and demand fail to match. By the very nature of the pro-rata mechanism, the resolution of supply-demand imbalances is unique, and the proportional rationing approach is fair. First, we consider markets with exogenous restrictions on the amounts each agent may buy and/or sell, deriving closed-form first-best consumption and characterizing how one-sided caps monotonically shift the Pareto price, while two-sided caps have an ambiguous effect. These results give the market-maker (or regulator) a tool for studying the impact of trading restrictions on price formation. Second, we study a leader-follower setting in which a regulator (the leader) sets the trading price by optimizing first its own objective, such as balancing social welfare against a target traded volume or a fairness objective such as Gini-type disparity measure across farmers' profitability, while farmers (the `followers') respond via pro-rata trading.  We further compare the proposed pro-rata approach to a family of asymmetric rationing schemes (seniority-based, excess-based, uniform, and mixed pro-rata rules) that trade off proportional fairness against protections for small or senior water-rights holders. Throughout, we illustrate the theoretical results with a numerical case study calibrated to a stylized four-farmer groundwater market. 
\\[0.5em]
\textsc{Keywords:} \ &  pro-rata rationing, groundwater price, Nash equilibrium, water rights, Pareto optimality, Gini index.   \\[0.5em]
			\textsc{MSC2020:} \ &  91B76, 91A10, 91B72 \\[0.5em]
			\textsc{JEL:} \ &  C72, Q25,  D47, Q58, D63   \\[1em]
            \hline
\end{tabular}
	}

\section{Introduction}
Groundwater---fresh water pumped from underground aquifers---is a critical source  in regions where precipitation and surface water resources are limited. For example, in California's Central Valley which supplies over 30\% of fruits and vegetables in US, groundwater accounts for 50-65\% of irrigation, especially during the rainless summer months. Historically, groundwater was extracted through small wells and viewed as part of the landowners' prerogative. Large-scale farming and accompanying industrial-scale pumping has led to a severe case of the tragedy of the commons: without a regulatory mechanism, the ``free'' groundwater resource is over-extracted by individual users. The impacts of the unsustainable status quo are manifold: dry wells, land subsidence, saltwater contamination, permanent aquifer compression, surface stream collapse, etc. 

To bring order to the ``free-for-all'' runaway pumping, management mechanisms have been adopted, aiming  to track and regulate groundwater usage. Existing groundwater management practices \cite{jakemanEtAl2016Book} are rooted in an allocation paradigm. First, pumping volumes are mandated to be measured \cite{maples2018leveraging} and strictly reconciled via water rights accounting. Second, regulators shape the overall pumping cap that modulates legal groundwater rights. Finally, regulators set rules governing pumping rights exchange and carryover, to supplement the legal governance laws. Extensive literature supports the resulting market approach as the most effective for sustainable management \cite{ayres2021environmental,AyresEtAl2021,HanakEtAl2023}.

To stabilize the aquifers, pumping rights are rationed, so that aggregate pumping is capped, often substantially below business-as-usual levels. To mitigate the resulting losses on farmers, market mechanisms permit groundwater trading, allowing  less-efficient (in terms of their individual shadow price of water) agents will sell their rights to the more-efficient agents that extract higher utility from water. The corresponding groundwater price acts as the signal that separates stakeholders into potential buyers and sellers.  In a fully market-driven set up, the groundwater price is also used as a clearing mechanism, leading to endogenous price formation as studied in the model of \cite{CialencoLudkovski2025}.
Revisiting that setup, we show that matching supply and demand is equivalent to the first-best central planner allocation (aka the Pareto equilibrium), and moreover maximizes the traded volume.

However, as shown in that reference, \textit{any} price can be associated to a Generalized Nash Equilibrium, and hence equating supply and demand, aka relying on the Pareto price $p^\circ$, is by no means the only plausible outcome. Indeed, there are multiple reasons why supply-demand imbalance might result. First, regulators often impose constraints on trading, for example capping how much water can be traded per agent, or targeting a particular traded volume. Second, regulators often directly intervene, meaning that the price is determined by additional factors beyond supply-demand matching. As two examples, regulators seek to prevent irreversible aquifer compression due to land subsidence, and to ensure environmental protection of local communities  that are  legally subordinated to the pumping rights of farmers \cite{ayres2021environmental}.  Third, stakeholders are often non treated equally, i.e., there is priority or preferential treatment rules which impact price formation. Due to all these reasons, it is common that groundwater trading occurs but supply and demand do not match. For more information about groundwater markets, we refer the interested reader to the dedicated resource\footnote{\url{https://cialenco.github.io/AquaStoch/}}. 

As soon as we have a market where $p \neq p^\circ$, the question of rationing re-appears. For example, suppose that regulators set a high groundwater price; then there would be more sellers than buyers and a scheme is needed to allocate the few rights that the ``most motivated'' buyers wish to purchase across all these sellers. The primary goal of this work is to study the mechanism of pro-rata which is the standard rule for partial matching of supply and demand and 
 has been used extensively in environmental rationing (timber and fishing quotas, pollution permits) and in financial markets (electronic trading of bonds). Proportional rationing underlying pro-rata generates complex effects in terms of gains and losses for the different stakeholders,  essential for understanding behavior of non-Pareto markets.

Beyond foundational analysis of pro-rata allocations, we investigate two setups that require it. First, we study \textit{constrained markets}, where trading volumes are restricted a priori. Hence, the unconstrained Pareto outcome is simply not feasible. Our key Theorems \ref{th:market-SB}-\ref{th:market-SB-Pareto-change} reveal how constraints on how many rights can be bought/sold impact the price and the consumption of each farmer. Second, we consider \textit{leader-follower markets} where the market-maker is an active regulator who explicitly sets groundwater prices. This is in contrast to  the Pareto setup, where the  state authority  acts like an ``invisible hand'', effectively optimizing the aggregate profits of the market participants, without criteria of their own. Instead, we view the regulator as optimizing $\sup_p [{\cal R}(p)-{\cal C}(p)]$, where ${\cal R}(p)$ is the ``revenue'' associated with groundwater price $p$, for example a weighted sum of farmers' gains and the value of environmental protection, and ${\cal C}(p)$ is the ``cost'', for instance the forgone trading gains relative  to the Pareto first-best market. Another motivation is a regulator that wishes to account for various collective externalities   and balancing the interests of all stakeholders, not just those with the legal rights \cite{bruno2024designing,kuwayama2013regulation}.  The interaction between the rationing rule and the market-maker preferences creates nontrivial effects that to the best of our knowledge have not been studied elsewhere. 

A market  where the price process $\{p(t)\}$ is first selected by the market-maker, and then farmers optimize their actions is an example of a leader-follower game \cite{insley2019climate}. 
The latter aspect links our work to literature on leader-follower oligopolies \cite{laffont1993theory}, including dominant firm with a competitive fringe and the work of Weitzman \cite{weitzman1974prices} on market regulation. Prior economic works on coupling this to rationing is \cite{weitzman1977price} and the recent \cite{ryan2022rationing} that specifically discusses rationing of groundwater.

In the broader class of \textit{principal-agent} games \cite{CvitanicZhang2012} there is a similar hierarchy where the regulator comes up with a market design, while strategically anticipating how stakeholders would react to such rules. In the second stage, pumpers interact among themselves, e.g., by trading their water rights.

The rest of the paper is organized as follows. Section \ref{sec:model} sets up our model and discusses properties of the Pareto price $p^\circ$. Section \ref{sec:prorata} defines the pro-rata rationing scheme and illustrates its implications for groundwater trading. Section \ref{sec:constrains} covers the case of constrained trading, whereby agents have caps on amounts of water traded. Section \ref{sec:Centralized-price} outlines objectives by the regulator; finally Section \ref{sec:conclude} concludes.

\section{Market model and problem formulation}\label{sec:model}
Our setup targets rural basins where the primary stakeholders are farmers and agricultural water districts, and the primary use of groundwater is for irrigation, i.e.~to grow agricultural products. Below we abstract from such details and use the generic economic terminology of agents and profit-loss functions.  Following the groundwater market model introduced in \cite{CialencoLudkovski2025}.

On a probability space $(\Omega, \sF, \bP)$, assume that $J$ economic agents trade within one basin groundwater rights among themselves. We use boldface to denote vectors, e.g.~$\boldsymbol{\varphi}=(\varphi_1,\ldots,\varphi_J)$. Fraktur letters  denote  sums of quantities across agents, e.g. $\mathfrak{W}=\sum_{j=1}^JW_j$.  

Each agent $j$  produces goods according to a profit-and-loss function $G_j:\Omega\times[0,\infty)\to\bR$, which maps \textit{water use} $C_j\geq 0$, in acre-feet (ac-ft) to profit (in dollars) $G_j(C_j)$. We assume that  the function $C_j \mapsto G_j(C_j)$ is increasing,  strictly concave and differentiable on $\mathcal C_j$. In particular, higher water use leads to higher profit. 

Agents also can trade water among themselves.
Groundwater trading is interpreted as individual groundwater users transferring their allocations, with no physical conveyance of groundwater between them. These trades are ``annual leases'' and only apply for the given period. 
Denote by $\psi_j$ the amount of \textit{water traded}, in ac-ft,  by agent $j$, with  convention that  $\psi_j>0$ means selling water, and $\psi_j<0$ buying water.  Let $p$ be the price  of traded water (in \$/ac-ft). Both $\psi_j, p$ could be random. Then, the $j$-th agent P\&L is given by 
\begin{align}\label{eq:L}
L_j(C,\psi, p) := G_j(C_j) + \psi_j \cdot p. 
\end{align}
The traded water amounts must balance out via the \textit{market clearing condition} 
\begin{equation}\label{eq:mrktClearPsi}
\sum_{j=1}^J \psi_j = 0. 
\end{equation}
The market-clearing condition is binding and imposed as a hard constraint, though it could be relaxed by allowing participants to opt out of trading.

Let $W_j$ be the total amount of water available to the agent $j$.
We impose the water budget constraints
\begin{equation}
    -\sum_{i\neq j} W_i  \leq C_j + \psi_j \leq  W_j. \label{eq:waterBudgetConstr}
\end{equation}

We also assume that $C_j \in \cC_j := [\underline{c}_j,  \overline{c}_j]$, for some fixed constants $0\leq \underline{c}_j < \overline{c}_j < \infty$. Such bounds on water consumption may arise either from physical limitations—reflecting the finiteness and bounded capacity of the aquifer—or from constraints imposed by production requirements. Additional constraints may be added, such as minimum and/or maximum water consumed or traded by each agent. 

Overall, agent $j=1,\ldots,J$ selects a control $\pi_j = \big(C_j, \psi_j\big)$ from the feasible action set 
$$
\cA_j(\bfW):= \mathcal{C}_j \times [-\sum_{i\neq j} W_j, {W}_j],
$$ 
which is a compact in $\bR^{2}$, for every allocation $\bfW$.  More generally, agents may employ \textit{mixed} (or randomized) strategies, by choosing a probability distribution over the action set $\cA_j$.

Additionally, we  introduce a fictitious player, called the \textit{price-setter}, whose payoff  is 
$L_{J+1}(\bfpi, p) : = p  \big(\sum_{j=1}^J  \psi_j\big)^2$, 
and who is choosing only $p$.
Without loss of generality, we assume that $p\in[\underline{p}, \bar p]$, for some fixed $0< \underline{p}<\overline{p} <\infty$. Indeed, under some natural growth condition on $G_j$, as proved below, for prices larger than a threshold $\bar p$, all agents will prefer to sell their water rights, hence nobody will buy, yielding the critical no-trade situation $\forall j,\ \psi_j=0$. Similarly, for small positive prices, all agents will prefer to buy. We denote by $\sS_{J+1}$ the set of adapted processes $p$ with values in $[\underline{p},\bar{p}]$, which is also the set of admissible strategy for agent $J+1$. Moreover, by analogy a mixed strategy for agent $J+1$ is a probability distribution on $[\underline{p},\bar{p}]$.

Farmer  $j$  maximizes the value function:
\begin{align}
& \max_{C_j,\psi_j} G_j(C_j) + p\psi_j,  \label{eq:model1-Criteria}\\
\textrm{s.t. } & \underline{c}_j \leq C_j \leq \bar{c}_j,  \label{eq:model1} \\
& \psi_j\leq W_j-C_j, \quad j=1, \dots J,  \label{eq:model1-constrpsi}\\
& \sum_{i=1}^{J} \psi_i=0. \label{eq:model1-markt-clear}
\end{align}

Without loss of generality, we consider\footnote{Generally speaking, one may argue that the lower bound $\underline{c}_j$ could be set to zero, allowing a farmer to sell all of her water rights. In this case, our framework requires additional technical assumptions that $G_j'$ are bounded from above.} $0 < \underline{c}_j<\overline{c}_j$. 

As is traditional in these types of markets, we assume that agents play a non-cooperative game, with the equilibrium determined as a Nash equilibrium, which we recall next.

\begin{definition}
 A feasible strategy $(\bm{\pi}^*, p^*)$ is called a \textit{Nash Equilibrium (NE)} to the problem \eqref{eq:model1-Criteria}-\eqref{eq:model1-markt-clear} if
 \[
\forall j = 1, \dots J \quad \forall \pi' _j \in \sS_j, \hspace{1mm} L_j(\pi_j^*, \pi_{-j}^*, p^*) \geq L_j(\pi'_j, \pi_{-j}^*, p^*).
\]
\end{definition} 

From the general theory of (stochastic) games \cite{BasarZaccour2018}, one can show that there exists a NE in the class of mixed strategies. Sometimes, when the constraints of one player depend on the strategy of others, the games are called Generalized Nash Equilibrium problem.  Usually, this equilibrium is not unique. In particular, as shown in \cite{CialencoLudkovski2025}, for every price $p>0$ there exists  a (pure) NE.  From both a computational and an economic perspective, it is desirable to refine the problem and develop a method that yields a unique, or at least clearly computable, Nash equilibrium.

\begin{definition}
 A feasible strategy $(\bm 
\pi^a, p^a)$ is called \textit{Pareto optimal} if there is no other feasible strategy $(\bm{\pi}' , p')$ such that
\[
\forall  j = 1, \dots, J, \hspace{1mm} L_j(\bm{\pi}' , p') \geq L_j(\bm\pi^a, p^a) \quad \text{and} \quad \exists i \hspace{1mm} \hspace{1mm} L_i(\bm{\pi}' , p') > L_i(\bm\pi^a, p^a).
\]
\end{definition}
 
\subsection{First-best consumption problem}
Central part of the analysis plays the desired optimal amount to consume and trade by each trader $j$ individually, by solving the so-called the first-best consumption problem \eqref{eq:model1-Criteria}-\eqref{eq:model1-constrpsi}. We denote its solution by  $(C_j^\circ(p), \psi_j^\circ(p))$, and correspondingly we call $C_j^\circ(p)$ the first-best consumption $\psi_j^\circ(p)$ the first-best traded amount by farmer $j$ for price $p$.  With Lagrangian 
\[
\cL = G_j(C_j) + p\psi_j + \theta_j (C_j - \underline{c}_j) + \beta_j(\overline{c}_j- C_j) + \eta_j (W_j-C_j-\psi_j)
\]
we get the KKT conditions, for some multipliers $\theta_j, \beta_j, \eta_j\geq0$,   
\begin{align}
     G_j'(C_j) + \theta_j -\beta_j -\eta_j & =0, \label{eq:KKT2-foc}\\
 p-\eta_j & = 0, \\
\theta_j(C_j-\underline{c}_j) & = 0, \label{eq:KKT2-lower}\\
\beta_j(\overline{c}_j -C_j) & = 0, \label{eq:KKT2-upper}\\ 
 \eta_j(W_j-C_j-\psi_j) & = 0, \label{eq:KKT2-psi}
\end{align}
for $C_j\in[\underline{c}_j, \overline{c}_j]$, and \eqref{eq:model1-constrpsi} satisfied. Since $p>0$, we have $\eta_j=p>0$, and thus \eqref{eq:KKT2-psi} is binding $\psi_j^\circ=W_j-C_j$, which means that the agent will always buy or sell the remaining amount of water, which agrees with the assumptions that the agent ignores the market clearing conditions. Thus, \eqref{eq:KKT2-foc} becomes
\begin{equation*} 
  G_j'(C_j) = p + \beta_j - \theta_j.   
\end{equation*}
If $\theta_j=\beta_j=0$, then \eqref{eq:KKT2-lower}-\eqref{eq:KKT2-upper} are not binding, and hence $\underline{c}_j< C_j <\overline{c}_j$, in which case $G_j'(C_j) = p$. Recall that $G_j$ is concave, hence $G_j'$ is decreasing on $\underline{c}_j< C_j <\overline{c}_j$, and thus $G_j'(C_j) = p$ will be satisfied, if and only if 
\begin{equation*}
G_j'(\overline{c}_j) < p < G_j'(\underline{c}_j) 
\end{equation*}
with optimizer $C_j^\circ(p) = (G_j')^{-1}(p)$. If $p\geq G_j'(\underline{c}_j)$, then $\beta_j=0$, $C_j^\circ=\underline{c}_j$, and $\theta_j = p- G_j'(\underline{c}_j) \geq 0$. 
Similarly, if $p\leq G_j'(\overline{c}_j)$, then $\theta_j=0$, $C_j^\circ=\overline{c}_j$, and $\beta_j = G_j'(\overline{c}_j) - p \geq 0.$ 
Putting, 
\[
\underline{p}_j := G_j'(\overline{c}_j), \quad  \overline{p}_j := G_j'(\underline{c}_j),
\]
and combining the above, we deduce 
\begin{equation}\label{eq:model1-cnsmpt-optimal}
C_j^\circ(p) = 
\begin{cases}
    \overline{c}_j, & p \leq  \underline{p}_j \\
    (G_j')^{-1}(p), &  \underline{p}_j< p < \overline{p}_j  \\  
    \underline{c}_j, & p \geq \overline{p}_j 
\end{cases},
\end{equation}
and 
\begin{equation}\label{eq:psi-circ}
\psi_j^{\circ}(p) = W_j-C_j^{\circ}(p) = 
\begin{cases}
    W_j-\overline{c}_j, & p \leq \underline{p}_j \\
    W_j-(G_j')^{-1}(p), & \underline{p}_j < p < \overline{p}_j\\
    W_j - \underline{c}_j, & p \geq  \overline{p}_j
\end{cases}. 
\end{equation}

\begin{remark}\label{rem:Ccirc-cont}    
Note that by construction, $C_j^\circ(\cdot)$ is a continuous function. More generally, by Berge's maximum theorem, this result remains true by assuming only concavity of $G_j$. 
\end{remark}

In what follows we use the notations
$$ 
\underline{p}=:\min_{j} \underline{p}_j, \qquad \overline{p}=:\max_{j} \overline{p}_j. 
$$
We remark that for $p\leq \underline{p}$, all agents will prefer to buy water, and for $p\geq \overline{p}$, sell water, and thus for prices outside the interval $[\underline{p}, \overline{p}]$ no water rights are traded. 

\begin{theorem}\label{th:Pareto-price}
    If $\sum_{i=1}^J \underline{c}_i < \sum_{i=1}^J W_i< \sum_{i=1}^J \overline{c}_i$, then there exists a price $p^a$ at which each farmer maximizes her payoff $L_j$, and the market clears. Moreover, $p^a$ is a Pareto optimal solution.  
    If additionally $\max_j \underline{p}_j < \min_j \overline{p}_j$, then $p^a$ is unique.  
\end{theorem}

\begin{proof} 
For $p>0$, denote 
$$
Z(p) := \sum_{i=1}^J C_i^\circ(p) - \sum_{i=1}^J W_i.
$$ 
In view of Remark~\ref{rem:Ccirc-cont}, $Z(\cdot)$ is continuous. 
By \eqref{eq:model1-cnsmpt-optimal}, we note that $C_j^\circ(p) \leq \underline{c}_j$, for $p\geq \overline{p}$, and for all $j=1\ldots,J$. 
Since $\sum_{i=1}^J \underline{c}_i < \sum_{i=1}^J W_i$, we have that $Z(p)<0$, for any $p\geq \overline{p}$. Similarly, for $p\leq \underline{p}$, we have that $Z(p)>0$. By intermediate value theorem, there exists $p^a$ such that $Z(p^a)=0$. Taking $\psi_j^\circ=W_j-C_j$, we obtain that the market clearing conditions \eqref{eq:model1-markt-clear} is satisfied. 

Denote by $\bfpi^a:=(\bfC^\circ(p^a), \boldsymbol{\psi}^\circ(p^a))$. Next we show that $(\bfpi^a,p^a)$ is Pareto optimal. We proceed by contradiction, and assume that $(\bfpi^a, p^a)$ is not Pareto optimal, that is there exists another feasible strategy $(\bm{\pi}' , p')$ such that
$$
L_j(\bm{\pi}' , p') \geq L_j(\bm{{\pi}^a}, p^a) \quad \forall j
$$
and there exists $k$ such that $L_k(\bm{\pi}' , p') > L_k(\bm{\pi}^a, p^a)$
Thanks to \eqref{eq:model1-markt-clear}, we have
$$
\sum_{i=1}^J L_i(\bfpi^a, p^a) = \sum_{i=1}^J G_i({C}^\circ_i), \quad \sum_{i=1}^J L_i(\boldsymbol{\pi}', p') = \sum_{i=1}^J G_i(C'_i). 
$$
From here, and the strict inequality for agent $k$, we deduce that 
\begin{eqnarray*}
   \sum_{i=1}^J L_i(\bm{\pi}' , p') > \sum_{i=1}^J L_i(\bm{{\pi}}^a, {p}^a),
\end{eqnarray*}
which implies that 
\begin{equation}\label{eq:parateo-strict}
   \sum_{i=1}^J G_i(C'_i) > \sum_{i=1}^J G_i({C}_i^\circ(p^a)).
\end{equation}
Since $C^\circ, \psi^\circ$ are maximizers, we have 
\[
G_j(C_j^\circ(p^a)) + p^a \psi_j^\circ(p^a) \geq G_j(C_j') + p^a \psi_j' \quad \forall j.
\]
Summing up the above, and using market clearing condition \eqref{eq:model1-markt-clear}, we obtain
$$
\sum_{j }G_j(C_j^\circ(p^a)) \geq \sum_{j} G_j(C_j'),
$$
which contradicts \eqref{eq:parateo-strict}.

To prove the uniqueness of $p^a$, note that the strict concavity of $G_j$ implies the strict monotonicity of $(G_j')^{-1}$, and hence the strict monotonicity of $C_j^\circ(p)$ on $(\underline{p}_j,\overline{p}_j)$. Consequently, under the assumption $ \max_i \underline{p}_i < \min_i \overline{p}_i$, it follows that $Z(p) = \sum_j (C_j^\circ(p) - W_j)$ is strictly monotone decreasing in $p$ on  $(\max_i \underline{p}_i, \min_i \overline{p}_i)$. Since by the existence part $Z(p)$ changes signs on $(\underline{p},\overline{p})\supset (\max_i \underline{p}_i, \min_i \overline{p}_i)$, we conclude it  admits a unique zero $p=p^a$ in this interval. This concludes the proof. 
\end{proof}

Next we prove some additional properties of the Pareto price $p^a$. 

\begin{proposition}
The maximum volume of water traded is attained at the Pareto price $p^a$.    
\end{proposition}
\begin{proof} 
Let us denote by $\cS(p)$ and $\cB(p)$ the set of sellers and buyers at price $p$. We recall that $C^\circ(p)$ is decreasing as function of price $p$. Thus, $W_j-C_j^\circ(p)$,  the desired amount of water to be traded by farmer $j$, is increasing in $p$.  Moreover, for $p_1<p_2$, we have  
\[
\cB(p_1) = \set{j \, :\, W_j-C_j^\circ(p_1) < 0 } \supset    \set{j \, :\, W_j-C_j^\circ(p_2) <0} = \cB(p_2),
\]
and correspondingly $\cS(p_1) \subset \cS(p_2)$. This implies that for $p>p^a$
\[
- \sum_{j\in \cB(p)} (W_j - C_j^\circ(p)) leq   - \sum_{j\in \cB(p^a) } (W_j - C_j^\circ(p^a)) = \sum_{j\in\cS(p^a)} (W_j - C_j^\circ(p^a)) \leq  \sum_{j\in\cS(p)} (W_j - C_j^\circ(p)) . 
\]
Thus, the amount of water traded at price $p$, equal to $-\sum_{j\in \cB(p)} (W_j - C_j^\circ(p))$, is smaller or equal than the amount traded at Pareto price $p^a$. Case $p<p^a$ is treated similarly. The proof is complete.   
\end{proof}

We recall that the \textit{social optimum} is the optimum chosen by the social planner whose sole objective is to maximize total social welfare $\sum_j L_j(C,\psi,p)$, formally solving 
\begin{align}\label{eq:model-social}
& \max_{C,\psi, p}  \sum_{j=1}^J L_j(C_j, \psi_j, p) = \max_{C_j,\psi_j,p} \sum_{j=1}^J  G_j(C), \\
\textrm{s.t. } \  &  \eqref{eq:model1}-\eqref{eq:model1-markt-clear}.
\end{align}

\begin{proposition}
\label{prop:SO-Pareto} Grant assumptions of of Theorem~\ref{th:Pareto-price}. Then, the social optimum coincides with the Pareto optimal solution $(\bfpi^a,p^a)$.
\end{proposition}

\begin{proof} 
The feasible set of the social planner's problem \eqref{eq:model-social} for decision variables $\psi_j, C_j$ is
$$
\cA = \Big\{ (C,\psi) \ :\ C_j\in[\underline c_j,\overline c_j],\ \psi_j\le W_j-C_j\ \ \forall j,\ \ \sum_{j=1}^J \psi_j = 0\Big\}.
$$
For $(C,\psi)\in\cA$, summing $\psi_j\le W_j-C_j$ over $j$ and using $\sum_j\psi_j=0$, we deduce
$$
\sum_{j=1}^J C_j \le \sum_{j=1}^J W_j.  
$$
We claim that at any social optimum $(C^{SO},\psi^{SO})$, the last inequality becomes equality. Indeed, suppose that 
\begin{equation}\label{eq:SO1}
\sum_j C_j^{SO} < \sum_j W_j. 
\end{equation}
Then, there exists an index $i$ such that $\psi_i^{SO} < W_i - C_i^{SO}$. Moreover, by the statement assumptions, $\sum_j W_j < \sum_j \overline c_j$, and thus there exists an index $k$ such that $C_k^{SO}<\overline{c}_k$, else all $C_j=\overline{c}_j$, and $\sum_j C_j^{SO} = \sum_j \overline c_j > \sum_j W_j$ contradicting \eqref{eq:SO1}.   
Fix $\varepsilon>0$ small enough (specified below) and define $C_k^\varepsilon = C_k^{SO}+\varepsilon$ and $C_j^\varepsilon=C_j^{SO}$ for $j\neq k$.

If  $i\neq k$, then define $\psi_i^\varepsilon=\psi_i^{SO}+\varepsilon$, $\psi_k^\varepsilon=\psi_k^{SO}-\varepsilon$, and keep all other unchanged $\psi_j^\varepsilon=\psi_j^{SO}$ for $j\notin\{i,k\}$. Choose $\varepsilon$ small enough so that 
$$
0<\varepsilon<\min\Big\{\overline c_k - C_k^{SO},\ W_i-C_i^{SO}-\psi_i^{SO}\Big\}.
$$
Clearly, $C_k^\varepsilon\le\overline c_k$. Moreover, since $i\neq k$, we also have $\psi_i^\varepsilon\leq W_i-C_i^{SO}=W_i-C_i^\varepsilon$, and  $\psi_k^\varepsilon = \psi_k^{SO}-\varepsilon \leq W_k-C_k^{SO}-\varepsilon = W_k-C_k^\varepsilon$. Since all other constraints are untouched, we obtain 
$$
\sum_j \psi_j^\varepsilon = \sum_j\psi_j^{SO} + \varepsilon - \varepsilon = 0,
$$
making $(C^\varepsilon,\psi^\varepsilon)\in\cA$ feasible.

If $i=k$, then take $\psi_j^\varepsilon=\psi_j^{SO}$ for all $j$, i.e. no change in strategy $\psi$, and thus the market clearing condition \eqref{eq:model1-markt-clear} remains satisfied by $\psi^{SO}$. Choose $\varepsilon$ such that
$$
0<\varepsilon<\min\Big\{\overline c_k-C_k^{SO},\ W_k-C_k^{SO}-\psi_k^{SO}\Big\}.
$$
Then $C_k^\varepsilon = C_k^{SO} + \varepsilon \leq \overline c_k$, and $\psi_k^{SO} \leq W_k-C_k^{SO}-\varepsilon = W_k - C_k^\varepsilon$. Thus, all constraints are satisfied, and in this case we again have $(C^\varepsilon, \psi^\varepsilon) \in \cA$. In general, the  strategy $(C^\varepsilon, \psi^\varepsilon) \in \cA$ is feasible.     

Recall that $G_j$ is strictly increasing, and combined with the above, we deduce
$$
\sum_{j} G_j(C_j^\varepsilon) = \sum_{j\ne k} G_j(C_j^{SO}) + G_k(C_k^{SO}+\varepsilon) \;>\; \sum_j G_j(C_j^{SO}),
$$
contradicting optimality of $(C^{SO},\psi^{SO})$, and hence \eqref{eq:SO1} must hold with equality
\begin{equation}\label{eq:SO2}
\sum_{j=1}^J C_j^{SO} = \sum_{j=1}^J W_j. 
\end{equation}

Let $d_j = (W_j - C_j^{SO}) - \psi_j^{SO}$. By the feasibility conditions $d_j\geq 0$.  Using \eqref{eq:SO2} and market clearing condition $\sum_j \psi_j^{SO}=0$, we get at once that $\sum_{j=1}^J d_j =0$, and since each $d_j\ge0$, we conclude that $d_j=0$ for every $j$, which implies that 
\begin{equation}\label{eq:SO2-1}
\psi_j^{SO} = W_j - C_j^{SO}, \qquad j=1,\dots,J.
\end{equation}
Thus, only $C_j, \psi_j$ such that $C_j+\psi_j=W_j$ can be included in the feasible set, which combined with the market clearing condition \eqref{eq:model1-markt-clear}, problem \eqref{eq:model-social} is equivalent to
\begin{align}
& \max_{C_j\in[\underline c_j,\overline c_j]} \sum_{j=1}^J G_j(C_j) \label{eq:SO3}\\
\textrm{s.t.} & \ \ \sum_{j=1}^J C_j = \sum_{j=1}^J W_j. \label{eq:SO3-1}
\end{align}
By strict concavity of $G_j$,  and using  $\sum_i\underline c_i<\sum_i W_i<\sum_i\overline c_i$, this is a concave optimization program, for which the KKT conditions are necessary and sufficient for global optimality. With multiplier $\lambda$ for the equality constraint and $\theta_j,\beta_j\ge0$ for the box constraints, stationarity gives, for each $j$,
\begin{equation}\label{eq:SO4}
G_j'(C_j) + \theta_j - \beta_j - \lambda = 0, \qquad \theta_j(C_j-\underline c_j)=0,\qquad \beta_j(\overline c_j-C_j)=0, 
\end{equation}
which coincides with the KKT system \eqref{eq:KKT2-foc}, \eqref{eq:KKT2-lower}, \eqref{eq:KKT2-upper} of the individual first-best problem with $\lambda$ in place of the market price. Hence, for any $(C^{SO},\lambda^{SO})$ solving \eqref{eq:SO4},
\begin{equation}\label{eq:SO5}
C_j^{SO} = C_j^\circ(\lambda^{SO}), \qquad j=1,\dots,J.
\end{equation}
The equality constraint \eqref{eq:SO3-1} implies that  $\sum_j W_j - \sum_j C_j^\circ(\lambda^{SO}) =0$. In the notation of Theorem~\ref{th:Pareto-price}, we have that $Z(\lambda^{SO})=0$. Hence, in view of the same theorem $\lambda^{SO}=p^a$, and by \eqref{eq:SO5}, $C^{SO} = \bfC^\circ(p^a) = \bfC^a$. Finally, by \eqref{eq:SO2-1},
$$
\psi_j^{SO} = W_j - C_j^{SO} = W_j - C_j^\circ(p^a) = \psi_j^\circ(p^a) = \psi_j^a.
$$
The proof is complete. 
\end{proof}

\section{A pro-rata approach to market imbalance}\label{sec:prorata}
We recall that $C_j^\circ= C_j^\circ(p)$ denotes the \textit{first-best water consumption} farmer $j$ would like to pump, for an exogenously fixed price $p$. Consequently, the desired traded quantity is $\psi_j^\circ = W_j-C_j^\circ$. If $\psi_j^\circ<0$, then it 
represents the farmer's $j$’ desired net bid (or demand) amount of water in the open market. Similarly, if  $\psi_j^\circ>0$, then it denotes the desired net ask (or supply) amount. We also recall that $C^\circ_j(p)$ is decreasing and continuous in $p$. 
Let  $\widetilde{p}_j=\min \set{p \ : \ C^\circ_j(p)=W_j}$, which can be seen as \textit{indifference price} at which the farmer's allocated water rights matches her first-best demand. If  $\underline{c}_j<W_j<\overline{c}_j$, then $\underline{p}_j< \tilde p_j < \overline{p}_j$.
Therefore, if $\underline{c}_j<W_j<\overline{c}_j$, for all $j$, then for 
$$
p \geq \widetilde{p}_u:=\max_j \widetilde{p}_j,
$$ 
no one is willing to buy water, resulting in no trade. Similarly, for 
\[
p\leq \widetilde{p}_l := \min_j \widetilde{p}_j,
\]
no agent is willing to sell, hence there are no trades. We also recall that regardless of the allocations $W_j$'s, no trades occur for any $p\leq \underline{p}$ or $p\geq \overline{p}$. Therefore, trades will happen only for 
$$
\underline{p} \vee \widetilde{p}_l \leq p \leq \widetilde{p}_u \wedge \overline{p}.
$$
Moreover, if we only assume that $\underline{c}_j\leq W_j$, that is,  each farmer has sufficient  initial water to meet their minimal requirements, then trades occur for 
\[
\widetilde{p}_l \leq p \leq \widetilde{p}_u \wedge \overline{p}.
\]
As before, let $\sS(p)$ and $\sB(p)$ denote the set of sellers and buyers respectively for a fixed price $p$, and denote by $\bfS(p)$, and $\bfB(p)$ the aggregate ask (or supply) and  bid (or demand) 
\begin{align*}
\bfS(p) := \sum_j (W_j-C_j^\circ(p))\1_{W_j-C_j^\circ(p)\geq0} =  \sum_{j\in\sS} (W_j-C_j^\circ(p)), \\ 
\bfB(p):= \sum_j (C_j^\circ(p)-W_j)\1_{W_j-C_j^\circ(p)\leq0} = \sum_{j\in\sB} (C_j^\circ(p)-W_j).   
\end{align*}
For simplicity of writing, when $p$ is fixed, and if no confusion arises, we simply write $\cB,\cS,\bfS,\bfB$. 

We postulate that the   market-maker will allocate the  water imbalance using a \textit{pro-rata algorithm}. That is, the water sold/bought by an agent is proportional to their individual ask (or bid) amount relative to the total supply(or demand).  Specifically:  

\medskip \noindent
\textit{Case A:} If $\bfS \leq  \bfB$, the supply is less than the demand, then, each agent  $j\in\sB$ on the buy side will consume   
\begin{align}
\widetilde{C}_{j} = \underbrace{W_{j}}_{\substack{\text{initial}\\ \text{allocation}}} + 
\underbrace{\frac{C_{j}^\circ - W_{j}}{\bfB}}_{\substack{\text{proportion of}\\ \text{total demand}}} 
\times \underbrace{\bfS}_{\substack{\text{total}\\ \text{supply}}}.
\end{align}
It is easy to check  that $\widetilde C_{j} \leq  C_j^\circ$ for $j\in\sB$, and for all $j\in\sS$. Hence, each buyer will not receive more than  desired first-best amount, and in particular the upper bound $\widetilde C_j \leq \bar c_j$, for $j\in\sB$, is satisfied.  In contrast, each seller $j\in\sS$ consumes her optimal amount $C_j^\circ$.  

\medskip\noindent
\textit{Case B:} If $\bfS>\bfB$, the supply is larger than the demand, then, for $j\in\sS$, each seller will sell 
\[
\underbrace{\frac{W_j-C_j^\circ}{\bfS}}_{\substack{\text{proportion of}\\ \text{total supply}}} \times \underbrace{\bfB}_{\substack{\text{total}\\ \text{demand}}}. 
\]
Since $\bfS>\bfB$, she will sell less than she desires, and thus, the unsold amount 
\[
W_j-C_j^\circ - \frac{W_j-C_j^\circ}{\bfS}\cdot \bfB = \frac{W_j-C_j^\circ}{\bfS}(\bfS-\bfB)
\]
will be used for additional production of goods, up to the maximum amount $\bar{c}_j$. Hence, the total amount consumed by the seller $j\in\sS$, is equal to 
\begin{align}
\widetilde{C}_{j} = \left(C_{j}^\circ + \frac{ W_j - C_j^\circ}{\bfS } \cdot 
( \bfS - \bfB) \right)\wedge \bar{c}_j. 
\end{align}
It is easy to show that $\widetilde{C}_j \leq W_j$. Note that in this case, in principle, it may happen that the unsold water by agent $j\in\sS$ is larger than the water she can use $\bar{c}_j-C_j^\circ$. In this case, this water will be simply left unused. This is counterintuitive, and the farmer  may be willing to sell it for less. However, if $W_j \leq \bar{c}_j$, then the boundary $\bar{c}_j$ will not be bypassed, and no unused water will be left. From practical point of view, it is also reasonable to assume that $W_j\leq \bar{c}_j$, which we postulate in what follows.  

\begin{theorem}
Under the pro-rata mechanism, every price $p$ is a Nash equilibrium.  
\end{theorem}

\begin{proof} Fix $p$. 
\textit{Case A:} If $\bfS\leq\bfB$, in view of the above, each agent \( j \in \sS \) adopts her first-best consumption $C_{j}^\circ(p)$, and every agent $ j' \in \sB$ consumes
\[
\tilde{C}_{j'} = W_{j'} + \frac{C_{j'}^\circ - W_{j'}}{\bfB } \bfS,  
\]
and \( \tilde{C}_{j'} \leq C_{j'}^\circ \). After pro-rata mechanism, agents $j\in \sS$ will not deviate since they already pump at the first-best level. 
Each agent $ j' \in \sB $ can not  unilaterally increase her profit, since that will require more water than $\tilde{C}_{j'}$ which was uniquely determined by the pro-rata allocation. If any agent \( j' \in \sB \) decides to consume less than \( \tilde{C}_{j'} \), again, since the objective function is a strictly increasing function of consumption, the profit will decrease, and thus not a Nash equilibrium. 

\noindent\textit{Case B:} If $\bfS >\bfB$, then all agents \( j' \in \sB \) will meet their first-best consumption, and each seller will follow the consumption strategy 
\[
\widetilde{C}_{j} = \left(C_{j}^\circ + \frac{ W_j - C_j^\circ}{\bfS } \cdot ( \bfS - \bfB) \right)\wedge \bar{c}_j. 
\]
Similar to Case A, buyers are already at first-best, and sellers cannot deviate unilaterally in their favor due to monotonicity of the objective function.  

This concludes the proof. 
\end{proof}

\begin{remark} Instead of using as the baseline for pro-rata the proportion of the demand or supply, one may consider an allocation based on the assigned water rights $W_j$. For example, if $\bfS>\bfB$, then each seller $j\in S$, gets to sell the amount   
\begin{align}\label{eq:prorata-wj}
\bar \cI_j:=\underbrace{\frac{W_j}{\sum_j(W_j\cdot \1_{W_j\geq C_j^\circ})}}_{\substack{\text{seller $j$’s share of}\\ \text{sellers’ water rights}}} \times \underbrace{\bfB}_{\substack{\text{total}\\ \text{demand}}}.  
\end{align}
However, \eqref{eq:prorata-wj} may lead to inconsistency, such as $\bar\cI_j \geq W_j-C_j^\circ$, meaning that a seller may get a larger share of the demand than she is willing to sell. A workaround is to have several allocation rounds, by eliminating those who sold all their water rights. While viable, this approach leads to a more nonlinear allocation rule and, from an economic perspective, is less desirable as it favors grandfathered water rights over market competitiveness. It may nonetheless be considered for comparison purposes.   
\end{remark}

Going forward, we use $V_j(p) := L_j(C_j, \psi_j, p)$ to denote the profits of Farmer $j$ given the price $p$. This is interpreted as mapping market prices to the equilibrium revenue of the farmers based on a known imbalance rationing mechanism, taken to be pro-rata by default.

\subsection{Illustrative Example}\label{sec:illustrative_example}

Let us consider a particular illustrative case of a groundwater market model with $J=4$ agents, profit functions of power type $G_j(C) = f_j C^{\alpha_j}$, $\underline{c}_j\leq C\leq \bar c_j$, and the  set of parameters given in Table~\ref{tab:water_params}. We also set $f_j=200$, for all $j$.

\begin{table}[h!]
\centering
\setlength{\tabcolsep}{10pt} 
\renewcommand{\arraystretch}{1.2} 
\begin{tabular}{ccccc}
\toprule
\cmidrule(lr){3-4} 
$j$ & $\alpha_j$ & $W_j$ & $\underline{c}_j$ & $\overline{c}_j$  
 \\
\midrule
1 & 0.5 & 200 & 60  & 250   \\
2 & 0.7 & 280 & 80  & 300  \\
3 & 0.6 & 260 & 100 & 350 \\
4 & 0.8 & 275 & 120 & 400  \\
\bottomrule
\end{tabular}
\caption{Model parameters for the case study in Section \ref{sec:illustrative_example}}
\label{tab:water_params}
\end{table}

By direct computation we find that 
\[
 \underline{p}= 6.324 \quad \widetilde{p}_l = 7.071, \quad \widetilde{p}_u = 52.029, \quad \bar p = 61.416.  
\]
The realized consumption functions $\widetilde C_j(p)$ and profits $V_j(p)$ under pro-rata allocation scheme (solid lines), as well as their individually desired optimal consumption $C_j^\circ(p)$ and corresponding first-best profits (dashed lines) as a function of $p$ are displayed in Figure~\ref{fig:pro-rata1}. 

\begin{figure}
    \centering
    \includegraphics[width=0.99\linewidth]{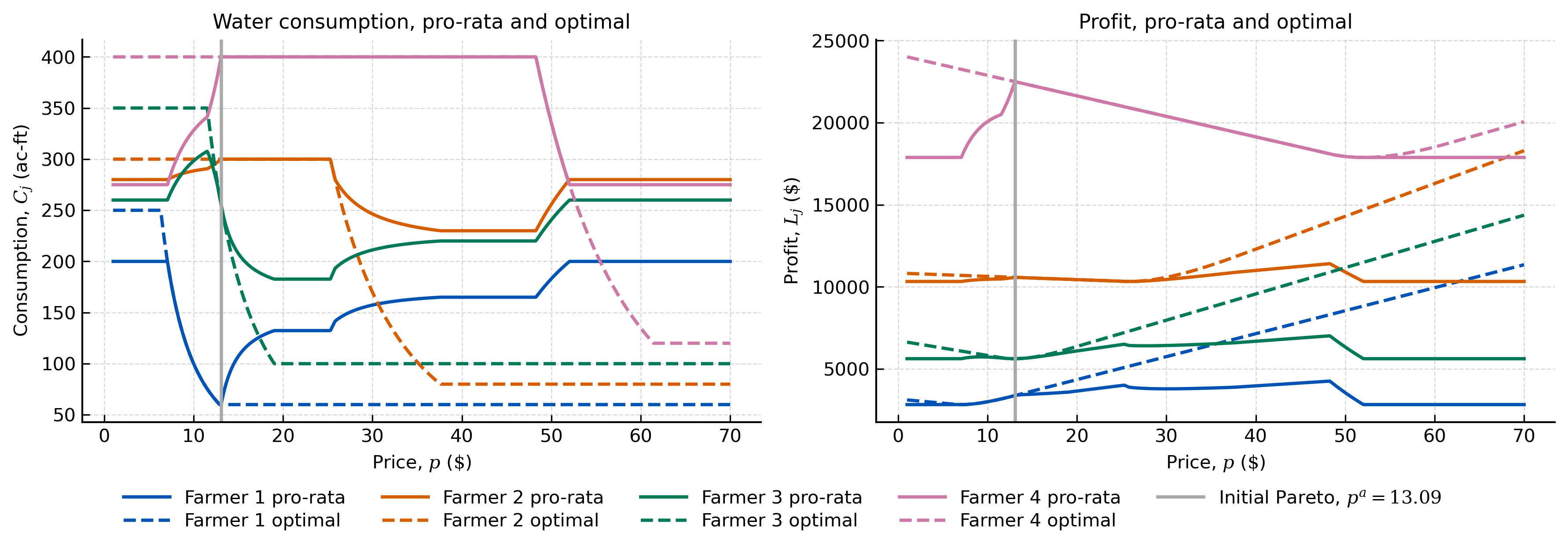}
\caption{Pro-rata (solid lines)  and first-best desired consumption (dashed lines) for the illustrative  example in Table \ref{tab:water_params}. Pareto price $p^a$ shown as gray vertical line. Left panel: water consumption $p \mapsto \widetilde C_j(p), p \mapsto C_j^\circ(p)$. Right panel:  Profit $p\mapsto V_j(p), \ p\mapsto L_j(C^\circ_j(p), \psi^\circ_j(p), p)$.}
    \label{fig:pro-rata1}
\end{figure}

For prices $p \le \widetilde{p}_l = 7.1$, all farmers prefer to purchase water, so each fully consumes their entire water entitlement and no trade occurs at these prices. The dashed consumption curves show that, at such low prices, farmers would like to consume more, but their consumption is constrained by the limited availability of water rights.

Within the trading window, $ \widetilde{p}_l\leq p \leq \widetilde{p}_u$, some farmers choose to buy while others sell. Farmer~1 (blue curves) is the first one willing to sell, while other farmers act as buyers, acquiring water from this seller. Thus, Farmer~1 achieves her first-best consumption, while there is not enough supply for buyers to reach their own desired optima. As the price $p$ increases, Farmer's~1 own consumption decreases, causing her to sell more water, and thus buyers’ consumptions increase. For Farmer~1, selling is more profitable than consuming at that range,  explaining the simultaneous increase in both the seller’s and buyers’ profits. At $p=11.15$  Farmer~3 (Green lines) also becomes a seller, and a similar dynamic unfolds between the 2 sellers (F1, F3) and the two buyers (F2, F4).

The solid gray  vertical  line marks the Pareto price, equal to the equilibrium social welfare price $p^a=13.09$, at which the aggregate demand  is equal to the aggregate supply ($\bfS=\bfB$), ensuring that each farmer achieves their first-best desired consumption at that price. 

For prices $p > p^a$, supply exceeds demand ($\bfS>\bfB$), buyers meet their first-best consumption levels, and sellers must consume more than their personal optimal. Hence, the desired optimal consumption and profit for buyers overlap after $p^a$ (solid and dashed curves coincide), whereas sellers’ profits diminish because consuming beyond their optimal level is less beneficial than selling. Finally, when $p > \widetilde{p}_u$, all farmers prefer to sell rather than buy, hence consuming only their water rights, and maintaining constant profit levels.

Overall, under pro-rata allocations (cf.~Figure \ref{fig:pro-rata1}), an agent’s consumption is no longer necessarily monotone in 
$p$, and profits exhibit more complex behavior compared to individually optimal profit levels. These local changes in groundwater prices may have unintuitive market effects, which we examine next.

 \section{Pro-rata with constraints on water traded }\label{sec:constrains}

Due to numerous unpriced externalities and the sensitive issue of maintaining a viable agricultural sector, most regulators do not allow unfettered trading within the groundwater market. In this section, we investigate market behaviors
when the volume of water that each agent may sell or buy is constrained. This is motivated by practical regulations imposed by district or state authorities \cite{heard2019sgma,edwardsRules}. A common mechanism are directional rules that impose one-way trading restrictions, where some agents are not allowed to buy, and others are not allowed to sell, generally linked to specific sensitive sub-zones, such as areas subject to saline intrusion. Such restrictions are intended to sustain groundwater levels, avoid drying out agricultural land, limit market power manipulation, and provide protection against droughts and other uncertainties. 
Constrained trading is likely to generate an imbalance between first-best water demand and supply and the feasible trading amounts and hence naturally calls for pro-rata or other allocations mechanisms.  Importantly, as we show below trading caps shift the Pareto price, which give the market-maker (or regulator) a tool for studying the impact of trading restrictions on price formation.

Motivated by above, we study several extensions of base or unrestricted model \eqref{eq:model1-Criteria}--\eqref{eq:model1-markt-clear}, by adding constraints on the traded amounts $\psi_j$. The resulting models are labeled Model~S, Model~B, and Model~SB, corresponding respectively to restrictions on sales, restrictions on purchases, and restrictions on both sales and purchases of water rights. Correspondingly, for each model  we use the notations $C_j^{\circ,k}, \psi_j^{\circ,k}, p^{a,k}$, with $k\in\set{S,B,SB}$.

We start with analysis of Model~SB, where the market maker (or regulator) imposes additional constraints that each farmer sells no more than a fraction $0 \le a_j \le 1$ and buys no more than $b_j \geq 0 $ of her total water rights $W_j$.   Hence, farmer $j$ solves the following modified problem
\begin{align}
& \max_{C_j,\psi_j} [G_j(C_j) + \psi_j p],
\label{eq:model-BS-Criteria}\\
\textrm{s.t. } \quad & \ \underline{c}_j \leq C_j \leq \bar{c}_j,  \label{eq:model-SB}\\
& -b_j W_j \leq \psi_j \leq a_jW_j \wedge  (W_j-C_j), \qquad j=1, \dots J, \label{eq:model-SB-Constrpsi}\\
& \sum_{i=1}^J \psi_i=0.
\label{eq:model-BS-markt-clear}
\end{align}
We note that if $a_j\geq  (W_j-\underline{c}_j)/W_j$, in particular $a_j=1$, then there are no restrictions on the sold amount. Correspondingly, for large enough $b_j$, for example for $b_j \geq (\sum_{i\neq j} (W_i - \underline{c_i}))/W_j$, there is no restriction on purchased amount for farmer $j$.  

Note that, in contrast to the unrestricted case -- where each agent may buy or sell an unlimited amount of water and can therefore always reach the lower and upper consumption bounds -- in the restricted trade market this is no longer automatic. To ensure the production bounds remain attainable, we assume throughout that
\begin{equation}\label{eq:SB1}
\underline{c}_j \leq W_j(1+b_j), \qquad W_j(1-a_j) \leq \overline{c}_j, \qquad \forall j, 
\end{equation}
that is, the water rights $W_j$ together with the maximum amount farmer $j$ is allowed to purchase (respectively sell) are sufficient to reach the lower bound $\underline{c}_j$ (respectively the upper bound $\overline{c}_j$) of the box constraint \eqref{eq:model-SB}. 

Let us define, what we will show later to be the effective bounds for the consumption $C_j$ 
\[
\underline{c}_j^{SB}= \underline{c}_j \vee W_j(1-a_j), \qquad  \overline{c}_j^{SB}= \overline{c}_j \wedge W_j(1+b_j). 
\]
Clearly $\emptyset \neq [\underline{c}_j^{SB}, \overline{c}_j^{SB}]\subset [\underline{c}_j, \overline{c}_j]$.

The next Theorem provides the first-best consumption quantities $C^{\circ, SB}$ under the above constraints for any groundwater price $p$.

\begin{theorem}\label{th:market-SB}
For any fixed price $p>0$, the first-best consumption in the market Model~SB has the form  
\begin{equation}\label{eq:modelSB-first-best}
C_j^{\circ, SB}(p) = 
\begin{cases}
    \overline{c}_j^{SB}, & p \leq  \underline{p}_{j}^{SB} \\
    (G_j')^{-1}(p), &  \underline{p}_j^{SB}< p < \overline{p}_j^{SB}  \\  
    \underline{c}_j^{SB}, & p \geq \overline{p}_{j}^{SB} 
\end{cases},
\end{equation}
where $\underline{p}_{j}^{SB} = G_j'(  \overline{c}_j^{SB})$, $\overline{p}_{j}^{SB} = G_j'(\underline{c}_j^{SB})$. 
Moreover, $\psi_j^{\circ,SB}(p)=W_j - C_j^{\circ,SB}(p)$. 
\end{theorem}

\begin{proof} 
For a fixed $C_j\in[\underline c_j,\overline c_j]$, we consider the sub-problem in $\psi_j$,
$$
\max_{\psi_j} \ p\,\psi_j \qquad \text{s.t.} \quad -b_jW_j \leq \psi_j \leq a_jW_j \wedge (W_j-C_j).
$$
Since $a_jW_j\ge 0\ge -b_jW_j$, the feasible set of this problem is well defined iff  
\begin{equation}\label{eq:SB2}
a_jW_j\wedge(W_j-C_j) \geq -b_jW_j \iff W_j-C_j \geq -b_jW_j \iff C_j \le W_j(1+b_j). 
\end{equation}
Thus for $ C_j \le W_j(1+b_j)$, since $p>0$, the objective $p\psi_j$ is maximized by taking $\psi_j$ as large as possible in side the feasible set, namely 
\begin{equation}\label{eq:SB2-2}
\psi_j^*(C_j) := a_jW_j \wedge (W_j-C_j). 
\end{equation}
We note that combining \eqref{eq:SB2} with the constraint $C_j\leq\overline c_j$, the set of $C_j$ for which some feasible $\psi_j$ exists is $C_j\in[\underline c_j,\overline c_j^{SB}]$. On this set, problem \eqref{eq:model-BS-Criteria}--\eqref{eq:model-SB-Constrpsi} reduces to
\begin{equation}\label{eq:SB3}
\max_{C_j\in[\underline c_j,\ \overline c_j^{SB}]} \ h_j(C_j),  
\end{equation}
where $h_j(C_j)= G_j(C_j) + p\big(a_jW_j\wedge(W_j-C_j)\big)$. By \eqref{eq:SB2-2}, we note that if $\psi_j^{*}(C_j)=W-ja_j$, then the farmer sold the maximum amount allowed, and thus $C_j< W_j - W_ja_j$. Else, if $C_j \geq W_j - W_ja_j$, then $\psi_j^*(C_j) = W_j-C_j$. Therefore, we write  $h_j$ as follows
$$
h_j(C_j) = 
\begin{cases}
G_j(C_j) + p\, a_jW_j, & C_j < W_j(1-a_j),\\
G_j(C_j) + p\,(W_j-C_j), & C_j \ge W_j(1-a_j).
\end{cases}
$$
The right and left limit of $h_j$ at $C_j = W_j(1-a_j)$ are the same, and thus $h_j$ is continuous. Recall that $G_j$ is strictly increasing, and thus $h_j(C_J)$ is strictly increasing on $[\underline{c}_j, W_j(1-a_j)]$. This implies that the maximum of $h_j$  over $[\underline c_j,\overline c_j^{SB}]$ is never attained at a point strictly inside $\big[\underline c_j, W_j(1-a_j)\big)$. Consequently, the maximization problem \eqref{eq:SB3} is equivalent to
$$
\max_{C_j \in [\underline c_j^{SB},\ \overline c_j^{SB}]}\ G_j(C_j) + p\,(W_j-C_j).  
$$
which is exactly the first-best problem \eqref{eq:model1-Criteria}--\eqref{eq:model1-constrpsi} with $[\underline c_j,\overline c_j]$ replaced by $[\underline c_j^{SB},\overline c_j^{SB}]$, and the assertion follows at once. 
\end{proof}

Next, we study how the Pareto optimal price changes due to restrictions on sales or purchases.  
\begin{theorem}\label{th:market-SB-Pareto-change} 
Assume that $\sum_{i=1}^J \underline{c}_i < \sum_{i=1}^J W_i< \sum_{i=1}^J \overline{c}_i$. Then: 
\begin{enumerate}[(i)]
    \item if there are only restrictions on sales, namely all $b_j$ are large enough, it holds true that $p^{a,S} \geq  p^{a}$;
    \item if there are only restrictions on purchases, namely all $a_j=1$, it holds true that $p^{a,B} \leq p^{a}$. 
\end{enumerate}
\end{theorem}
\begin{proof} We focus on Model S with restrictions only on sales, and  assuming $b_j$ are all large enough. We recall that for the original, unconstrained model, $\psi^{\circ}$ is given by \eqref{eq:psi-circ}, while for Model~S, in view of Theorem~\ref{th:market-SB}, we have 
\begin{align}
\psi_j^{\circ,S}(p) = 
\begin{cases}
    W_j -  \overline{c}_j, & p \leq \underline{p}_j \\
    W_j-(G_j')^{-1}(p), & \underline{p}_j < p < \overline{p}_j^{SB}\\
    W_j - \underline{c}_j^{SB} & p \geq  \overline{p}_{j}^{SB}.
\end{cases}
\end{align}
Since $\underline{c}_j^{SB}\geq \underline{c}_j$, then $\psi_j^{\circ,S} \leq \psi_j^\circ$, and thus $\sum_j \psi_j^{\circ,S}(p) \leq \sum_j\psi_j^{\circ}(p)$. Moreover, both $\sum_j\psi_j^{\circ}(p)$ and $\sum_j \psi_j^{\circ,S}(p)$ are continuous and increasing functions of $p$. By Theorem~\ref{th:Pareto-price}, the Pareto prices are prices at which the market clears, namely satisfying \eqref{eq:model1-markt-clear}, \eqref{eq:model-BS-markt-clear}, or equivalently $p^a$ is the  zero of $\sum_i \psi_i^{\circ}(p)$, and $p^{a,S}$ is the zero of $\sum_i \psi_i^{\circ,S}(p)$.  Overall, we deduce that $\sum_j \psi_j^{\circ,S}(p^{a})\leq0$, which implies that $p^{a} \leq p^{a,S}$. 

Part (ii) is proved similarly, first noting that $\psi_j^{\circ,B} \geq \psi_j^\circ$. The proof is complete. 
\end{proof}

Besides outright caps on amount of water to buy/sell, several other restrictions have been proposed or considered, including i) spatially-dependent caps; ii) caps linked to farm size. For the first case, some areas of the basin may be designated as too sensitive for trading (for example due to environmental effects such as stream depletion) and hence are given additional restrictions compared to other, less sensitive areas. On a related note, there may be restrictions of moving rights from one sub-basin to another, e.g.~a cap of the form $\sum_{j \in E} \psi_j \le \bar{\Psi}$ for a subset of farmers $E$.  For example, stakeholders in the Goulburn-Murray basin in Victoria, Australia are restricted to transfer rights from deep confined aquifers (Zone 2 in Upper Ovens Catchment) to unconfined alluvial fan aquifer in Zone 1 \cite{VictoriaWaterRegister}. We leave such cases for future research. For case ii), several advocacy groups specifically target small farms, with some proposal including preferential/different trading rules for farmers with low $W_j$. The justification is that small farms bear the brunt of labor and community impact of agricultural shifts (e.g., potential land fallowing) and hence need separate protections compared to large landowners.

\subsection{Capped Sales of Water}\label{sec:Restricted-sales}

Restricting the volume  of water any stakeholder can sell is practiced in some markets in order to ensure that land is not completely fallowed and that agricultural use is maintained. For example, the rules of the Edwards Aquifer Authority (EAA) in Texas stipulate that  ``Owners of qualified irrigated land in this region were allocated two acre-feet of water per acre by the Edwards Aquifer Authority, based on proven historical use. The first or “top” acre-foot can be leased or sold for any end use. The second or “bottom” acre-foot must remain with the land in agricultural use.'' \cite{edwards2004}
This implies that $a_j = 0.5$ since $W_j = 2$ and $W_j-\psi_j > 1$. The Fox Canyon GSA in California restricts all sales to be at most 100\% of the annual allocation $W_j$ which is a meaningful cap in the context of multi-period markets with carryover provisions. Fox Canyon furthermore restricts all purchases (but still allows sales) by farmers located in one of the designated Special Management Areas \cite{heard2019sgma}, in order to make sure the pumping volume in ecologically fragile coastal zones can not increase beyond the allocation. 

\begin{figure}[!htb]
    \centering
    \includegraphics[width=0.99\linewidth]{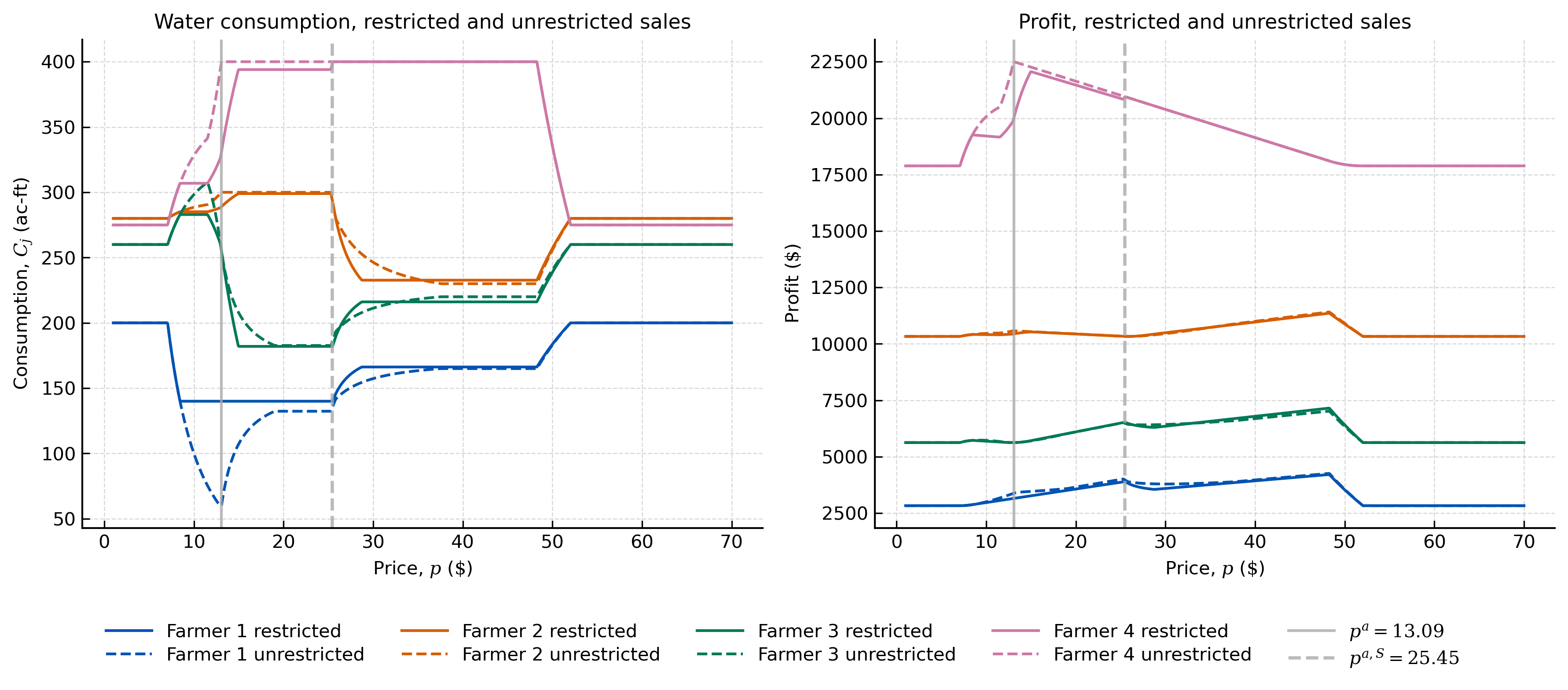}
\caption{Water consumption $C_j$ (left panel), Profits $V_j$ (right panel), and the Pareto prices $p^a$ for market models with restricted sales (solid lines) and unrestricted sales (dashed line).}
\label{fig:rest_sales}
\end{figure}

Figure \ref{fig:rest_sales} illustrates the impact of such restricted sales market. We re-use the case study and model parameters from Section~\ref{sec:illustrative_example}, imposing a cap of any farmer able to sell at most 20\% of their rights, i.e.~$a_j=0.2$ for all $j$. We show the consumption and profit $V_j(p)$ of each farmer as a function of $p$, along with their Pareto optimal prices $p^{\circ,k}$ (dashed vertical lines). The Pareto price in the unrestricted market $p^{a}=13.09$ is lower than the price in the restricted market $p^{a, S}=25.45$, illustrating Theorem \ref{th:market-SB-Pareto-change}.

For the primary seller (Farmer 1, blue), the 20\% cap on sales implies that restricted consumption increases compared to the unrestricted baseline,  $C^S \ge C$. This result stems from the pro-rata redistribution mechanism: any surplus exceeding the sales limit is added to the collective unsold total and redistributed in proportion to each agent's contribution. Given that selling is more profitable than consumption, the constraints also reduce the profits of Farmer 1,  as confirmed by the right panel of Figure \ref{fig:rest_sales}. The pink Farmer 4 never sells, hence the cap reduces their consumption and profits. However, for the intermediate Farmers 2 and 3, the cap affects them ambiguously, sometimes increasing and other times decreasing their consumption, with a positive or negative impact on profits.

In Figure \ref{fig:salesconstr},  we show the aggregate consumption $\sum_j C^{\circ,k}_j$ for both markets. The black horizontal line corresponds to the total water rights $\sum_jW_j=1015$.
Aggregate consumption in the restricted market is greater than or equal to that of the unrestricted market as the sales cap forces farmers to consume a minimum of $W_j(1-a_j)$, effectively creating a floor on demand even when the unrestricted first-best consumption would be lower.

\begin{figure}[htb!]
    \centering    \includegraphics[width=0.6\linewidth]{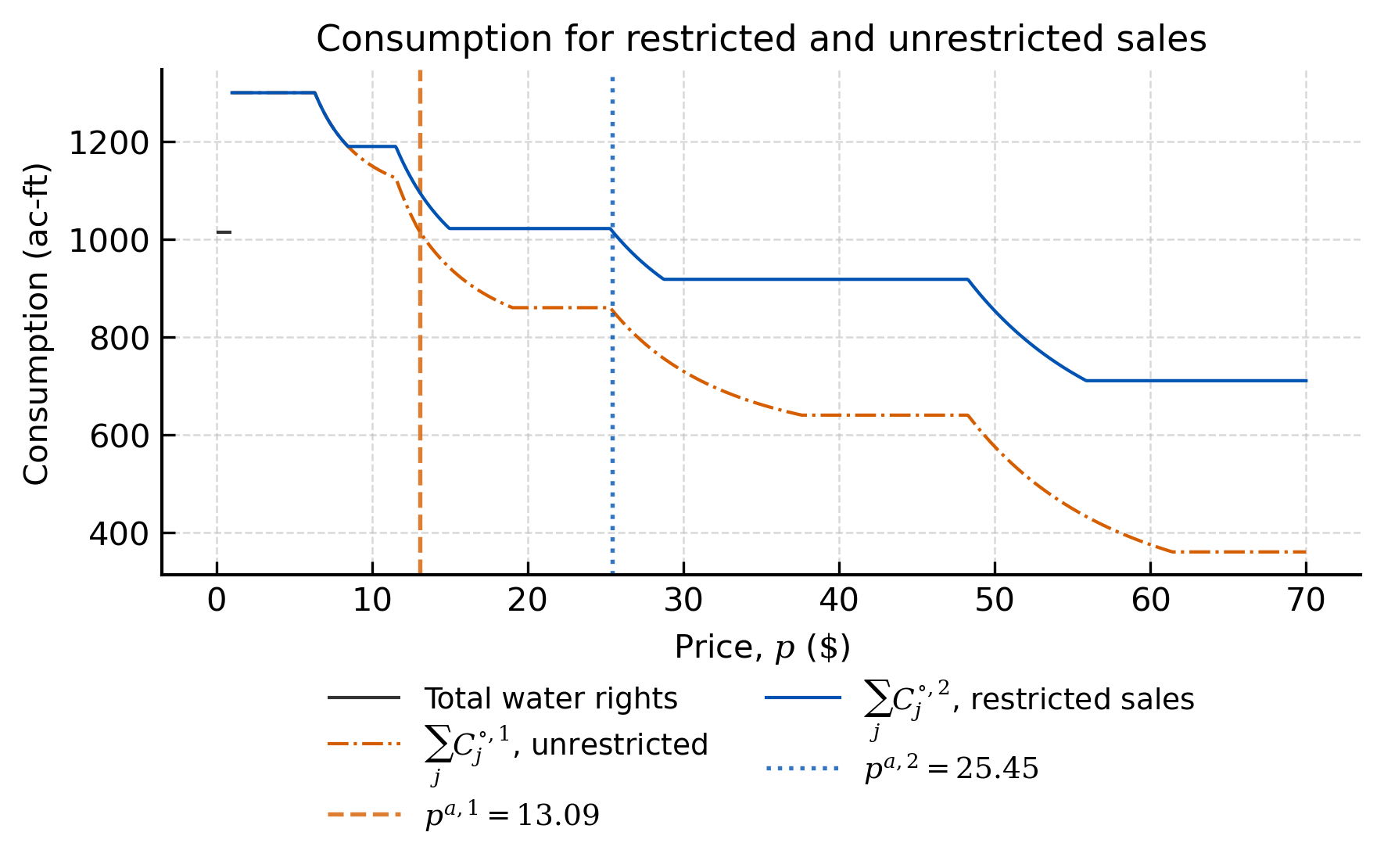}
    \caption{Aggregate optimal water consumption  and Pareto prices for unrestricted sales (orange lines) and restricted sales (blue lines) with restriction trading proportion$a_j=0.2$. Solid black line corresponds to the total water rights, $\mathfrak{W} = \$1015$. 
    }
    \label{fig:salesconstr}
\end{figure}

\subsection{Joint Effect of Restrictions}\label{sec:restrct-both}

To illustrate the interaction and the changing impact of different constraints, we next compare the following 3 situations to the baseline: (S) capped sales $a_j = 0.2$; (B) capped purchases $b_j = 0.2$; (SB) combined cap on both selling and buying rights, $a_j = b_j = 0.2$. These cases cover all the settings of Theorem \ref{th:market-SB} and are shown in Figure \ref{fig:All-restricted-farmer1}. The first case of capped sales (dotted red line) is exactly the case study from the previous section. When agents are prohibited from buying more than 20\% of their initial water rights (dashed blue), buyer consumption in the restricted market is bounded above by unrestricted consumption levels, leading to a corresponding reduction in profit. On the other hand, sellers experience involuntary increases in consumption due to the cap on market demand, realizing profits that are at most equal to their unrestricted counterparts. These dynamics are detailed in Figure \ref{fig:All-restricted-farmer1} for Farmer 1. Similar discussion holds for the bi-restricted market, see the green curve in Figure \ref{fig:All-restricted-farmer1}. In contrast to the previous two cases, $\sum_j \psi_j^{\circ,SB}(p) $ could be either bigger or smaller than $\sum_j \psi_j^{\circ}(p)$, which implies that the Pareto price $p^{a,SB}$ cannot be compared to $p^{\circ}$. 
Note that at intermediate prices, any constraint lowers the profits of Farmer 1. As expected, the dual restriction on both sales and purchases leads to largest loss compared to baseline, and in particular supersedes losses from only purchase caps or only sale caps. However, the losses between the latter two are non-comparable. Moreover, when prices are very low ($p<10$) or very high $(p> 50)$ the constraints do not bind.

\begin{figure}[!htb]
    \centering
    \includegraphics[width=0.99\linewidth]{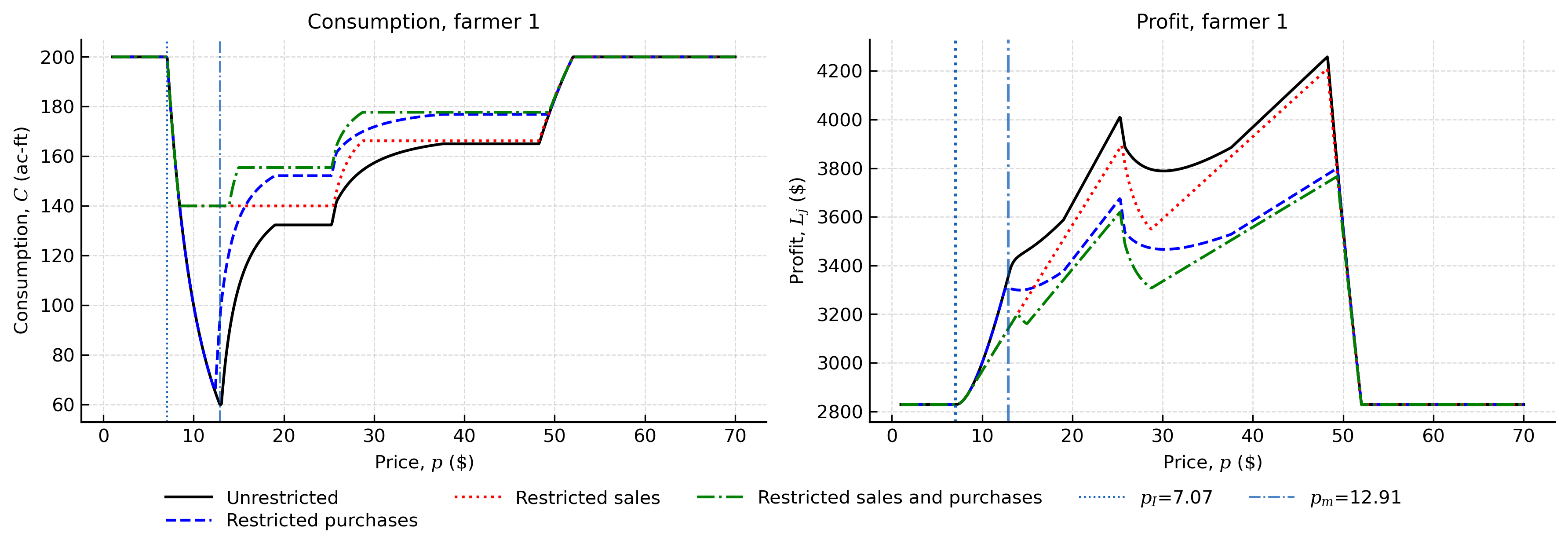}
    \caption{Impact of trading restrictions on Farmer 1 across the 4 cases $k \in\{S,B,SB\}$ in Section \ref{sec:restrct-both}. Consumption $C^k_1(p)$ (\textit{left panel}) and Profit $V^k_1(p)$ (\textit{right panel}) as a function of price $p$.}
    \label{fig:All-restricted-farmer1}
\end{figure}

Figure \ref{fig:price-a-b} shows how the Pareto clearing prices are driven by the joint choice of trading caps $a_j, b_j$. Recall that $a_j =1, b_j \gg 0$ corresponds to the unconstrained baseline. Moreover, per Theorem \ref{th:market-SB-Pareto-change}, prices increase in $a_j$ and decrease in $b_j$, so highest prices are when $a_j$ is large and $b_j$ is small (strict purchase caps) and lowest prices occur when $a_j$ is small and $b_j$ is large (strict sale caps).  By jointly tuning the sale and purchase caps $(a_j, b_j)$, the market-maker can effectively steer the equilibrium price to any target level, offering an indirect but powerful lever for price formation, assuming that the market trades at the Pareto price with imbalance settled through a pro-rata mechanism. The heatmap in particular makes clear that the effect is highly non-linear: tight caps on either side of the market lock in a small range of high Pareto prices, while as expected the price sensitivity to the caps flattens out once both $a_j, b_j$ exceed roughly $0.3$–$0.4$, closely matching the unconstrained Pareto price. 

\begin{figure}[!htb]
    \centering
    \includegraphics[width=0.6\linewidth]{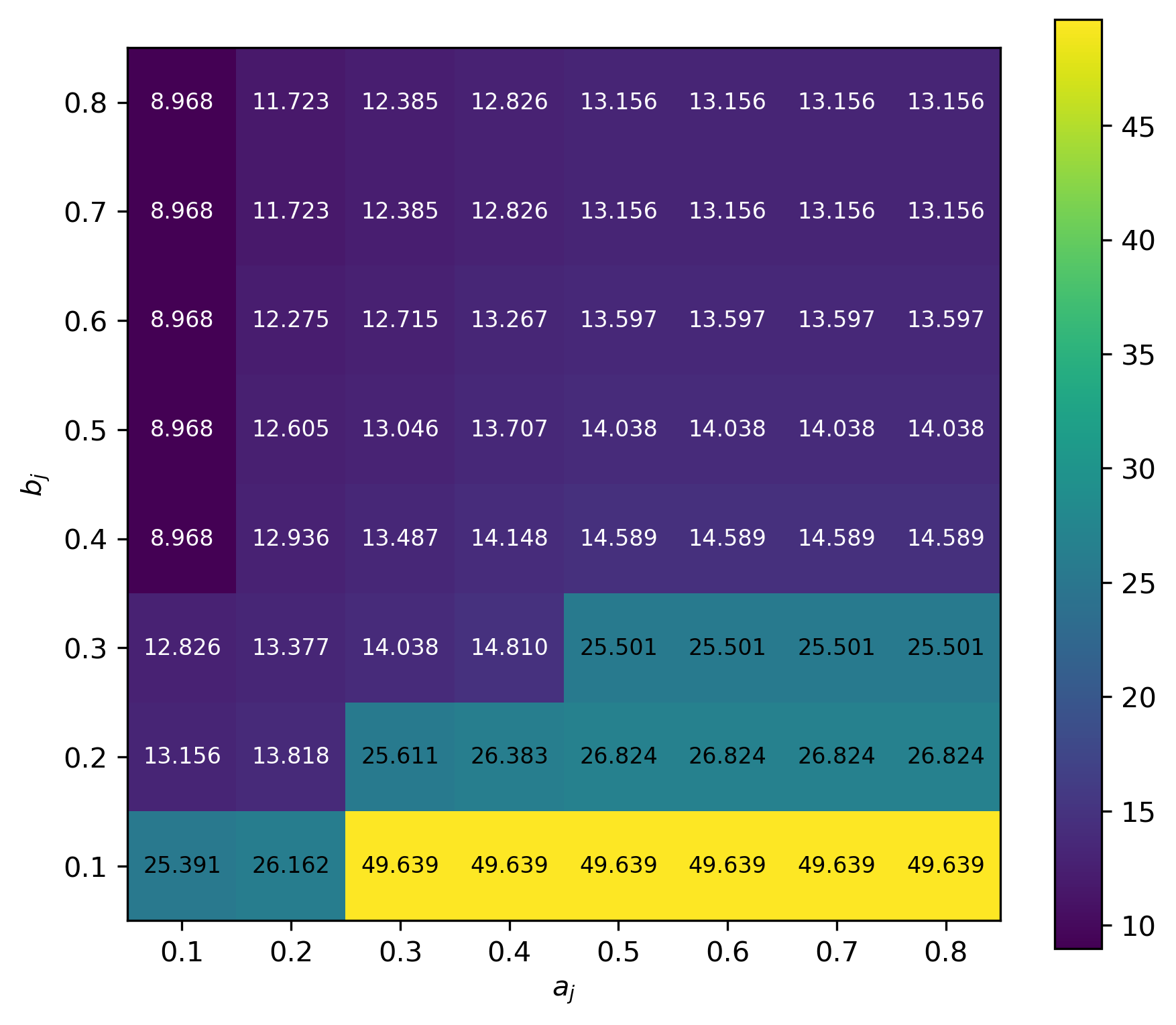}
    \caption{Pareto prices for different values of trading thresholds $a_j$ (max fraction of water sold) and $b_j$ (maximum fraction of water bought). The unconstrained case is in the upper right corner. }
    \label{fig:price-a-b}
\end{figure}

\section{Equilibrium under centralized price formation}\label{sec:Centralized-price}

In the base Pareto setup, the market-maker or state authority  acts like an ``invisible hand'', effectively optimizing the aggregate profits of the market participants (which we have shown is equivalent to maximizing trading volume), and does not have criteria of their own. In the more realistic setting, the market maker is a regulator who has additional objectives of their own, for example related to environmental protection or other externality not fully priced in by the farmers. 
Such goals can be embedded in the original problem by assuming that the market-maker picks the price $p$ either exogenously, or through an optimization criterion. The market-maker precommits to using the pro-rata mechanism so the agents know in advance how the market imbalance will be settled. As we argue below, this two stage procedure can be viewed as a leader-follower game.

\subsection{Targeting traded water volume}\label{sec:target_water}

Rather than imposing restrictions on individual traded volumes (see Section~\ref{sec:constrains}), a regulator may choose to limit the \textit{total volume of water traded} in the groundwater market. Such aggregate constraints can be motivated by the desire to stabilize the market, or to control  resource sustainability, particularly during periods of severe or uncertain droughts. An excessive amount of traded water during a short period of time may lead to abrupt shifts in groundwater use, increased pressure on vulnerable aquifers, or increase  price volatility. Capping the aggregate traded volume allows the regulator to  moderate collective responses, reduce systemic risk, and promote a more orderly adjustment of allocations over time leading to long-term market resilience.

We recall that the total amount of traded water $\sum_{j\in\sS} \psi_j(p)$, increases for $p \in(\widetilde{p}_l,p^a)$, and decreases for $p>p^a$, with maximum volume of $144.99$ ac-ft traded at the Pareto price $p^a$. Thus, the market-maker, instead of imposing an additional market constraint for all participants, can directly choose a price $p(A)$ that matches the desired amount of traded water $A$ (see Figure~\ref{fig:water_sold}, left panel) by solving 
\begin{equation}\label{eq:water_traded_rest}
    \sum_{j\in\sS} \psi_j(p) = A, \quad p\in[\widetilde{p}_l, p^a].  
\end{equation}
Clearly, for every $A$ strictly smaller than the maximum amount of water traded, there exist at least one price $p^m(A)<p^a$ and at least one price $p^M(A)>p^a$ satisfying \eqref{eq:water_traded_rest}. Consequently, the market maker, acting as leader, must also impose a restriction on the traded price relative to the Pareto price $p^a$. Choosing a smaller price---effectively $p^m(A)$---favors sellers of water, i.e., less efficient agents, whereas trading at $p^M(A)$ favors agents that use water more efficiently, as measured by profit per ac-ft of water used. While choosing $p^M(A)$ makes more sense from a competitive standpoint, regulators typically would opt for smaller price $p^m(A)$ instead, for two reasons. First, farmers tend to be wary of trading altogether, and lower prices help alleviate this reluctance; second, less efficient farmers are usually also the smaller ones, and favoring them is often consistent with the regulator's distributional objectives.

\begin{figure}[h!]
    \centering
    \begin{subfigure}[b]{0.49\linewidth}
        \includegraphics[width=.95\linewidth]{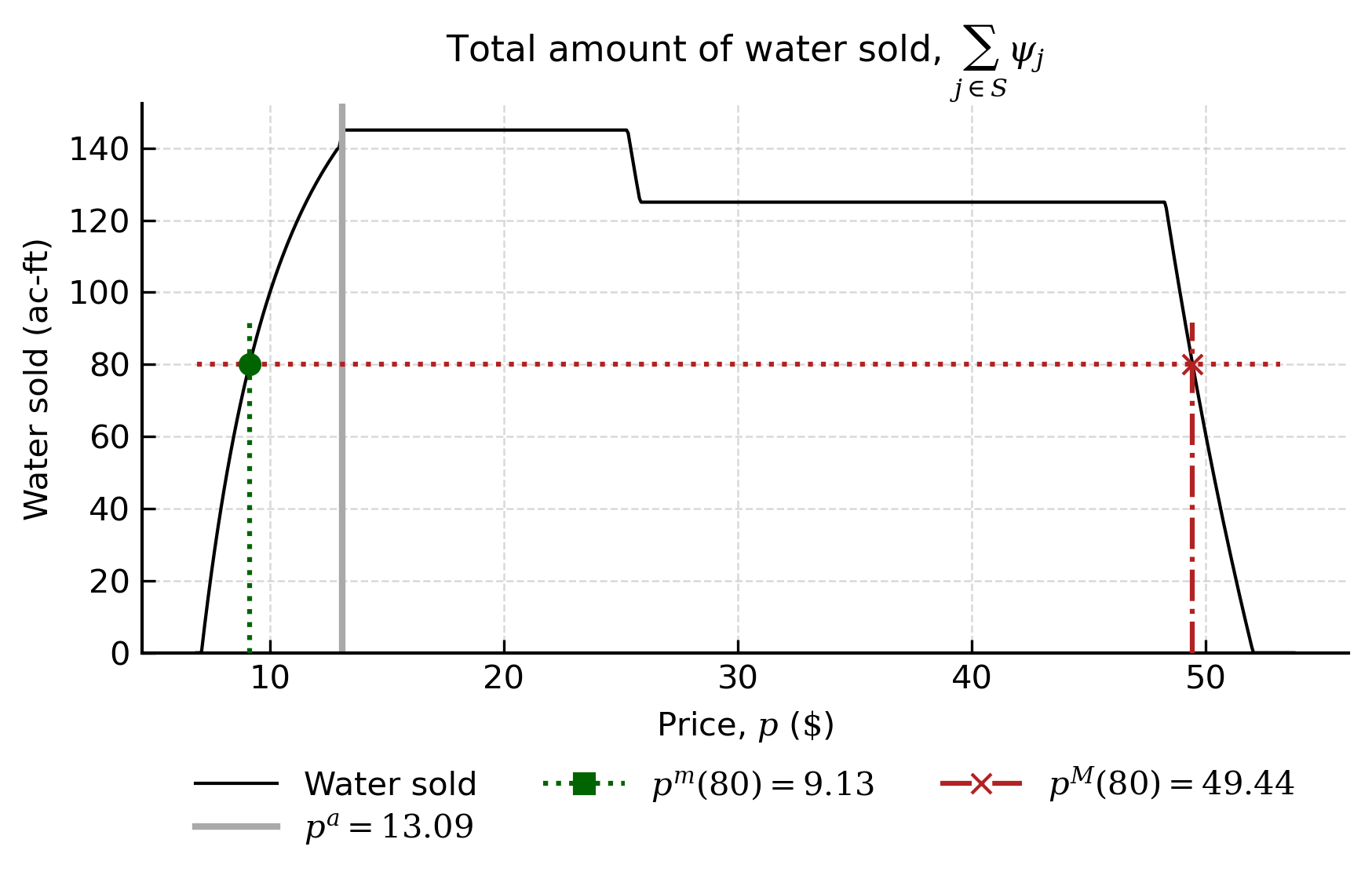}
    \end{subfigure}
    \hfill
    \begin{subfigure}[b]{0.49\linewidth}
        \centering
        \includegraphics[width=0.95\linewidth]{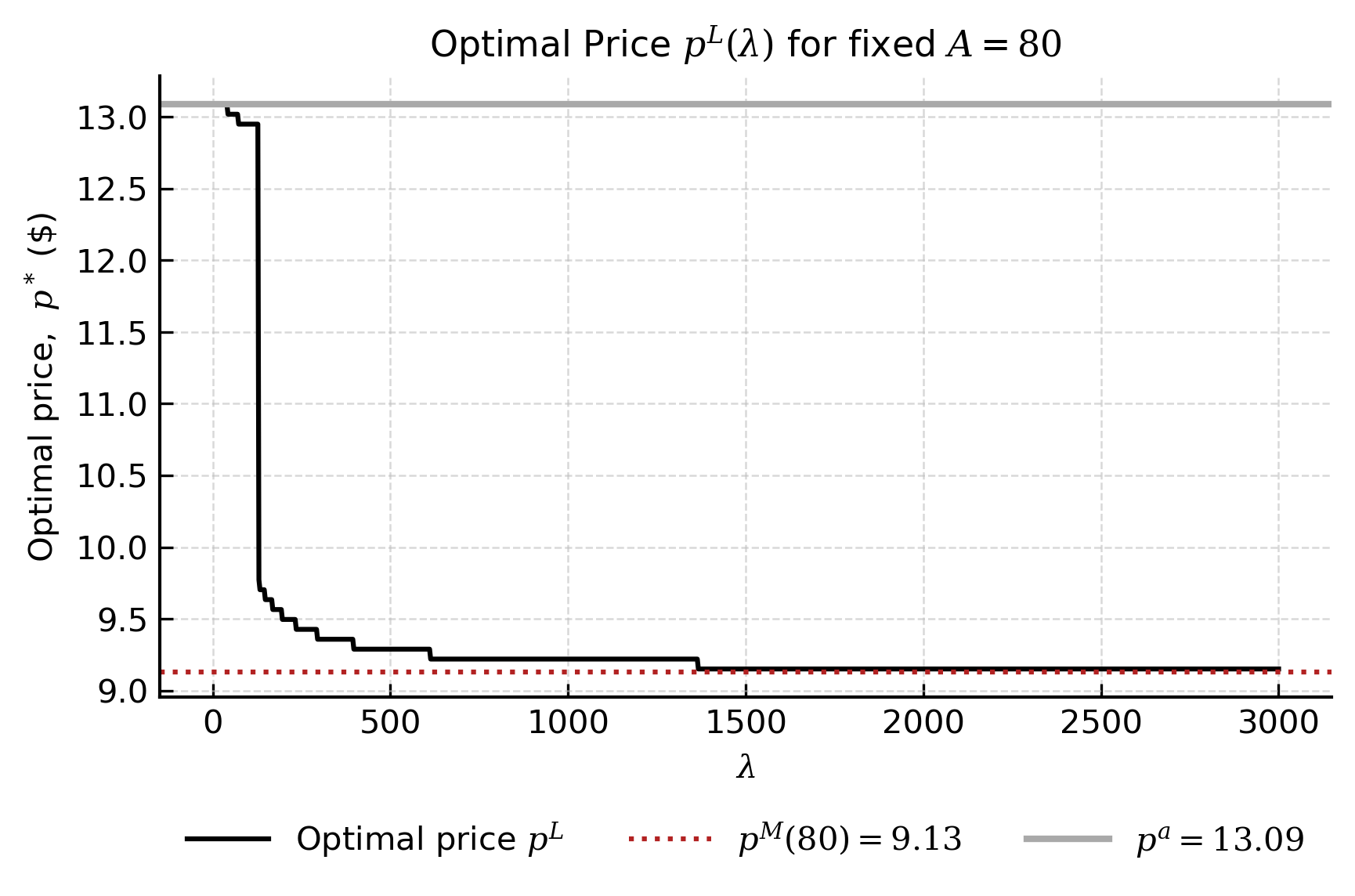}
    \end{subfigure}
    \caption{{Left: Total volume of water sold as a function of price $p$. Right: Optimal price $p^L$ in \eqref{eq:MM-objective} as function of penalty parameter $\lambda$. Gray lines correspond to Pareto price $p^a=13.09$, dotted red lines show $p^M(80)=9.13$ for target traded volume $A=80$.}}
        \label{fig:water_sold}
\end{figure}

We note that the uniqueness of $p^m$ or $p^M$ is not guaranteed and depends on the model parameters. In our example, for instance, the amount of water sold exhibits two flat regions for $p\geq p^a$, as shown in Figure~\ref{fig:water_sold} (left panel). This can be understood by inspecting Figure~\ref{fig:pro-rata1} (left): for prices immediately above $p^a$, Farmers~2 (orange) and~4 (pink) are buyers, both reaching their upper consumption bound $\overline{c}_j$, while Farmers~1 (blue) and~3 (green) are sellers, jointly supplying the entire demand, which itself remains constant. Under pro-rata allocation, the only change for Farmers~1 and~3 is the split of this fixed total between them.

Alternatively, the market-maker could act as a social planner who maximizes the social welfare and controls  the total amount traded through a soft penalty by considering the following objective  
\begin{equation} \label{eq:MM-objective}
\Phi^L(p, \lambda, A) = \sum_{j=1}^J \big( G_j(C_j(p)) + \psi_j(p) p \big)
+ \lambda \left( 1 - \frac{\sum_{j \in \sS} \psi_j(p)}{2A} \right) \sum_{j \in \sS} \psi_j(p),
\end{equation}
where $C_j(p), \psi_j(p)$ are the realized optimal consumptions and traded amount by agent $j$, under pro-rata mechanism, and $\lambda, A>0$ are some  fixed parameters.  The parameter $A$ as before should be viewed as the target volume of  water sold, and $\lambda$ as a strength or scaling parameter. As discussed above, we restrict prices to $p\leq p^a$, and the market-maker, or leader, solves the following problem 
\begin{equation}\label{eq:MM-problem}
p^L = \argmax_{p\leq p^a} \Phi^L(p, \lambda, A). 
\end{equation}
In turn, the farmers, or followers, given the optimal $p^L$ will solve for a Nash equilibrium with pro-rata allocation scheme.  
In view of Proposition~\ref{prop:SO-Pareto}, the first term in the objective \eqref{eq:MM-objective} achieves its maximum at the Pareto price $p^a$, at which we also have the largest amount of water sold $\sum_{j\in\cS} \psi_j^\circ(p^a)$. The second term achieves its maximum when $\sum_{j\in\sS} \psi_j(p)=A$. Thus, the market-maker, through appropriate choice of the parameter $\lambda$ (see Figure~\ref{fig:water_sold}, right panel), chooses the traded price by balancing between social welfare and target amount. For fixed $A$ and sufficiently large $\lambda$, as expected, the leader's chosen optimum price $p^L$ approaches $p^m(A)$. It should be noted that due to the particular choice of the penalty term in \eqref{eq:MM-objective}, for smaller $\lambda$ there is a regime shift with the optimum $p^L$ jumping to social optimum $p^a$ governed by the first term in the criterion. We refer to Appendix~\ref{sec:appendix_MM} for a detailed discussion of this effect.

\subsection{Optimizing the disparity in water-use profitability}\label{sec:disparity}
In a market with heterogeneous agents, farmers differ in their production technologies, water requirements, and outside options, and thus a single uniform price $p$ may affect them differently. As discussed earlier, some agents will find the price highly favorable relative to their marginal returns on water used, while others will not. This asymmetry makes it natural to seek a price that minimizes the resulting `disparity' in marginal profits across farmers. We start by defining the profit per unit of water used, or \textit{efficiency ratio}, as 
\begin{equation*}
\rho_j(p)=
    \begin{cases}
        \frac{V_j(p)}{W_j(p)+|\psi_j(p)|}, &  j \in \sB \\
        \frac{V_j(p)}{C_j(p)+\psi_j(p)}, &  j \in \sS
    \end{cases}
\end{equation*}
displayed for our numerical example in Figure~\ref{fig:gini} (left). As expected, Farmer 1 has lowest index, while Farmer~4 is the most efficient. 

A canonical way to measure market fairness, or disparity, is to use the Gini index, defined as 
\begin{align}\label{eq:gini}
G(p) = \frac{\sum_{i,j=1}^J |\rho_i(p) - \rho_j(p)|}{2(J-1) \sum_{i=1}^J \rho_i(p)}.  
\end{align}
The extreme values are defined correspondingly as $p^g = \argmin_{p} G(p)$, $p^G = \argmax_{p} G(p)$. In Figure~\ref{fig:gini} (right), we present the Gini coefficient as function of groundwater price $p$. The largest disparity is achieved at the threshold trading prices $\widetilde{p}_l,\widetilde{p}_u$; for prices smaller and larger than that there is no trading, and the Gini index  is the largest. The smallest Gini index is achieved at $p^g=48.24$, which is significantly different from the Pareto price $p^a=13.09$. For this particular set of model parameters, it happens that Farmer~4 has lowest efficiency coefficient, while all other agents have the largest one.  Comparing these two extremes offers a useful bracketing of outcomes: rather than looking only at a single equilibrium or efficiency-maximizing price, we can characterize the range of distributional outcomes that a price mechanism can plausibly produce, and assess how far a given policy price sits from either extreme.

\begin{figure}[h!]
    \centering
   \begin{subfigure}[b]{0.49\linewidth}
        \centering
        \includegraphics[width=\linewidth]{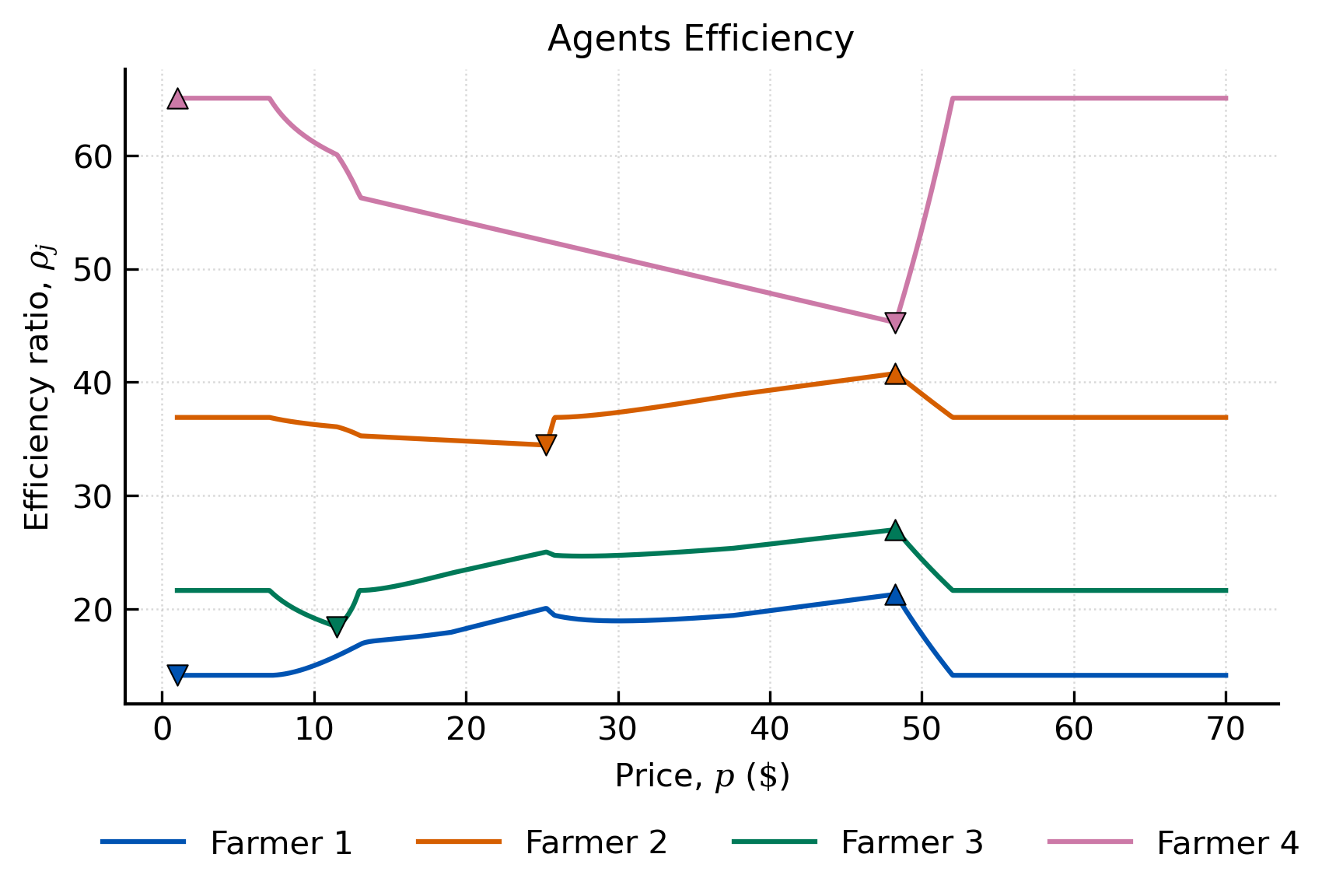}
    \end{subfigure}
    \hfill
    \begin{subfigure}[b]{0.49\linewidth}
        \centering
        \includegraphics[width=\linewidth]{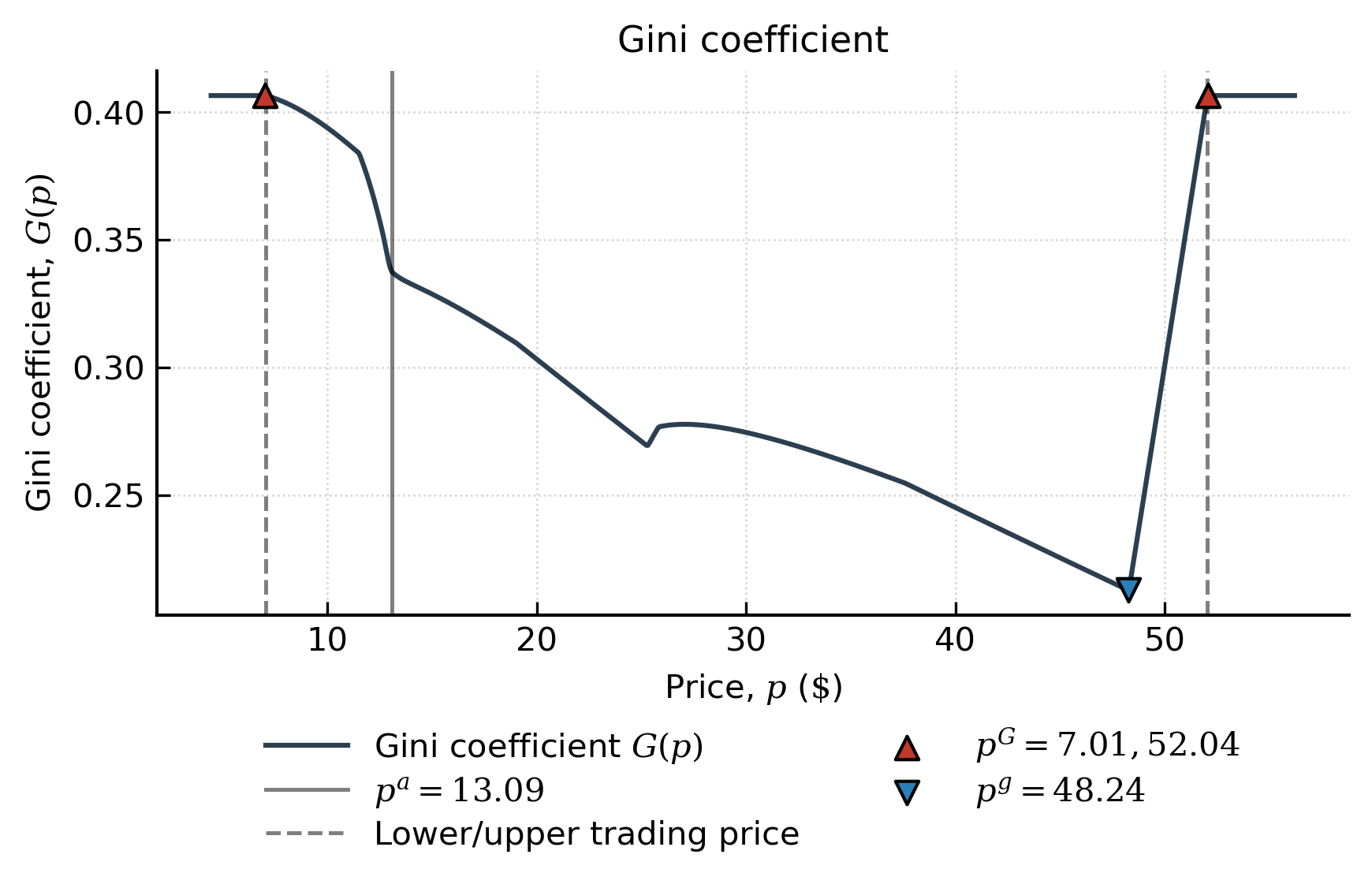}
    \end{subfigure}
    \caption{Left: Agents' efficiency coefficient (or profit per unit of water used) $\rho_j(p)$ as a function of $p$ for each $j=1,\ldots, 4$. Markers placed at largest and smallest value of each $\rho_j$. Right:  Gini coefficient $G(p)$. Dashed vertical lines correspond to traded price bounds $[\underline{p}, \bar{p}]$.}
    \label{fig:gini}
\end{figure}

Alternatively, the market-maker can measure disparity as cumulative pairwise (squared) distance between agents efficiency, defined as 
\begin{align}\label{eq:D-disparity}
D(p)=\sum_{i,j} |\rho_i(p)-\rho_j(p)|^2,
\end{align}
which can be equivalently written in terms of statistical variance of $\rho$, as $D(p) = 2J^2\Var(\rho(p))$. This measure yields similar results\footnote{Generally speaking, the extreme values of the Gini index $G(p)$ and disparity measure $D(p)$ do not coincide, and depend on particular model primitives.}
to the Gini index, and we omit discussing them here.

\subsection{Alternative Pro-rata Rules}

Beyond directly setting the market price of traded groundwater rights, the regulator has the complementary lever of prescribing the priority rule for prorating the imbalance. Rather than using the baseline pro-rata mechanism  under which all farmers are treated equally according to their buying/selling volumes, in this section we discuss  alternative schemes. 

Rather than resolving trading imbalance proportionally, there is a variety of asymmetric allocation schemes that preferentially put some stakeholders ahead of others. We refer to \cite{moulin2000priority,gomez2020agricultural} for an overview of the use of such schemes in environmental economics and water management more broadly. We note that in our setup, we must account for the original rights $W_j$ granted to each agent, as well as their minimum and maximum consumption bounds. 

The simplest asymmetric scheme is \emph{priority} based, whereby
agents are classified into priority classes and imbalance is resolved sequentially by class. That is, the needs of the agents with the highest priority are met first and, once fully satisfied, the remaining groundwater trades are allocated to the following agents according to a decreasing priority order criterion. If there are multiple stakeholders within a priority class, then those can be allocated by pro-rata, or by a further intra-class priority order. 

As an example, the Prior Appropriation doctrine ranks stakeholders chronologically according to the ``first in time first in right'' (also known as fully sequential allocation) and is often used for surface water rationing.

\paragraph{Prior Appropriation Seniority.} For groundwater rights, one common sorting is based on the farmers' legal water rights. For example, the Prior Appropriation doctrine ranks stakeholders chronologically according to the ``first in time first in right'' (also known as fully sequential allocation) and is standard for surface water rationing. Taking into account the assigned water rights $W_j$, the market maker may then apply either a Senior-first (SF) or Junior-first (JF) scheme, whereby the imbalance is allocated according to the original volume $W_j$. Thus, under Senior-first (respectively Junior-first) sorting, farmers with the largest (respectively smallest) water rights get to trade their full desired volumes first, followed sequentially by the remaining farmers, in a manner similar to a first-come first-serve allocation rule. A Junior-first prioritization may be beneficial to protect smaller farms with small $W_j$ allowing them to buy water (assuming $p< p^a$ and hence $\bfS < \bfB$) ahead of the larger farms whose demand may remain partially unfilled.

\paragraph{Uniform Rationing.}
Under the uniform rationing (UR) scheme, also known as Excess Gains, the imbalance is divided equally among market participants, up to their desired level.  Namely, if $\bfB<\bfS$, available buy volume $\bfB$  is first divided equally among the $|\sS|$ sellers, with each seller trading $(W_j-C_j^\circ(p))\wedge \bfB/|\sS|$.
Because of the latter upper bound, some farmers may reach their first-best sell volume in the first round, leaving some remaining unallocated buy volume. Hence, a  second round may be needed where any remaining buying interest is allocated equally among the  remaining sellers. These rounds are continued until all buying volume is spent. The case $\bfS<\bfB$ is treated analogously. 
UR tends to benefit buyers or sellers that have small imbalances $|W_j-C_j^\circ(p)|$ and get to fill them in the early rounds, while farmers with large imbalances are penalized.

Figure \ref{fig:OSU-all} in the Appendix illustrates the various versions of above schemes.

\paragraph{Mixed Pro-Rata (MPR).} As an intermediate scheme to SF/JF, farmers could be grouped into a few tiers, such that farmers in different tiers are treated differently according to the priority order criterion, while farmers classified in the same priority tier are treated under the axiom of equal treatment of equals, namely via proportional (pro-rata) or uniform rationing. For example, farmers could be exogenously divided (e.g. based on their land rights) into the senior Tier I and the junior Tier II. Market imbalances are first resolved within Tier I using pro-rata; any remaining trading opportunities trickle down to Tier II, allocated again on a pro-rata basis. With this scheme, senior farmers benefit from preferential access while receiving proportional treatment among their Tier I peers. Observe that  at the Pareto price, where aggregate supply equals aggregate demand, the MPR allocation coincides with the standard market outcome. Away from the Pareto benchmark, Tier I farmers may obtain higher profits at the expense of Tier II farmers. If the imbalance at price $p$ solely involves farmers from the same Tier,  the resulting allocation again coincides  with the standard pro-rata rule.

Figure \ref{fig:mixed-prorata} illustrates mixed pro-rata assuming that 
Farmers 1, 4 are in  Tier I (senior) and Farmers 2, 3 are in Tier II (junior). The resulting behavior is as follows:
\begin{enumerate}
\item For $\widetilde{p}_l =7.07 < p < 12.97$, the three buyers are Farmers 2, 3 and  4, with the latter getting strict priority and hence higher profit compared to regular pro-rata, while Farmers 2 and 3 are unable to buy (and so end up with lower profits). 
\item For $12.97 <p < p^a$, the two buyers are Farmers 2 and 3 who are in the same Tier II allocated proportionally. 
\item For $p> p^a=13.09$ the imbalance flips towards sellers. First for $p<25.84$, the sellers are Farmer 1 and Farmer 3 with Farmer 1 getting  priority to sell, while Farmer 3 is stuck with a higher-than-desired consumption. 
\item For $p \in [25.86,\widetilde{p}_u]$, $52.04=\widetilde{p}_u$ the only buyer is Farmer 4 (who buys the max she wants on this interval) and the sellers are Farmers 1, 2, and 3 who need to be rationed. Farmer 1 gets priority and sells everything while Farmers 2,3 are blocked from selling due to being junior. So the result is that only Tier 1 farmers trade with each other.
\item For $p>47.5$, Farmer 4 buys less so that Farmer 1 is getting to sell only a portion of her desired volume. 
\end{enumerate}

As can be seen in Figure \ref{fig:mixed-prorata}, the net result is that under MPR the Tier I farmers 1 and 4 get higher profits than under pro-rata, while Tier II farmers 2 and 3 get less; however for instance for $p=20$, the profits of Farmer 2 and 4 are unaffected by the tiering, and only the profits of Farmers 1, 3 shift.

\begin{figure}
    \centering
    \includegraphics[width=0.95\linewidth]{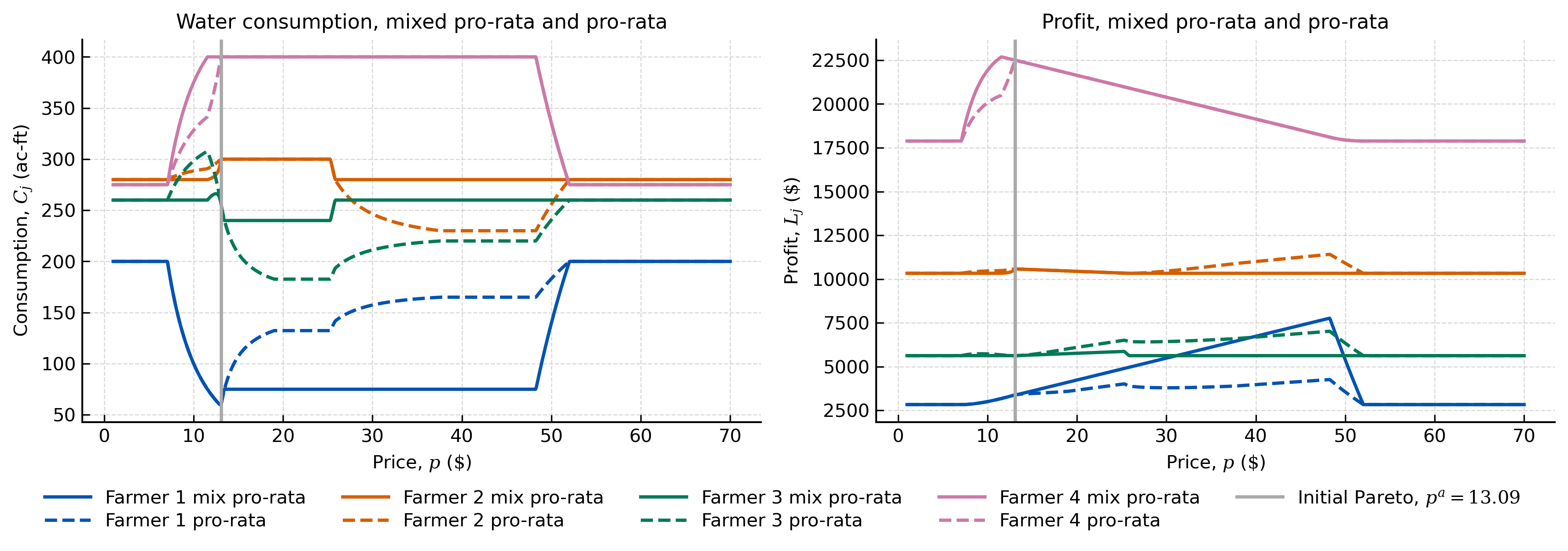}
    \caption{Mixed pro-rata allocation using market with 4 farmers and the model parameters from Section~\ref{sec:illustrative_example}. Farmers 1 and 2 are in senior Tier I that has priority over Farmers 3 and 4 (Tier II). \textit{Left:} consumption volumes of each farmer; \textit{right:} profits of each farmer. }
    \label{fig:mixed-prorata}
\end{figure}

Many other rules are possible. As one final idea, we mention a soft priority scheme where farmers are assigned weights  indicating their relative priority. Each buy/sell claim is multiplied by the weight assigned to the agent holding it, and allocation shares are calculated following pro-rata after this rescaling, while respecting consumption bounds and not going over first-best targets. Such weighted pro-rata can partially inflate the prioritization of some (say small) farmers.

\section{Concluding remarks}\label{sec:conclude}

We have introduced and analyzed a pro-rata rationing framework for groundwater markets that provides a principled response to the many practical situations---such as trading constraints, regulatory intervention, or priority trading rules---that lead to the selection of a NE different from the Pareto optimum. Building on the groundwater market model of \cite{CialencoLudkovski2025}, we have shown that the Pareto price $p^a$ simultaneously equates total supply and demand, solves the social planner's problem, and maximizes the traded volume. Any other price, while remaining a NE, creates an imbalance between bid and ask water volumes, leaving infinitely many possible allocations of the traded quantity. The proposed pro-rata framework resolves this indeterminacy by providing a unique and proportionally fair allocation of the traded water. Moreover, because every price remains a NE, it offers a natural and tractable mechanism for analyzing groundwater markets whenever they are driven away from the Pareto benchmark.

Deviations from the original Pareto price are discussed along two broad directions. \textbf{First}, we consider regulators who cap the volumes individual agents may trade, and show that restrictions on sales alone push the Pareto price up, while restrictions on purchases alone push it down. Combining both types of caps yields an ambiguous, model-dependent effect, as our numerical illustrations confirm. \textbf{Second}, we adopt a leader-follower framework in which the regulator acts as the leader and sets the price according to its own objective, while farmers act as followers and, under pro-rata settlement of market imbalances, choose their optimal trading strategies. For the regulator, motivated by practical considerations, we consider objectives that target a prescribed total traded volume (or balance social welfare against that target through a penalty parameter), or minimize the Gini index of farmers' efficiency ratios as a proxy for reducing disparities in market outcomes.

Additionally, we examined other rationing rules and how they reshape distributional outcomes. We compared the pro-rata mechanism with seniority-based (SF/JF),  uniform, and mixed allocation schemes. No single alternative to pro-rata dominates uniformly, underscoring that the choice of rationing mechanism is itself a substantive policy decision, distinct from and complementary to the choice of trading restrictions or the price-setting rule.

Several directions remain for future work. On the theoretical side, extending the framework to multi-period trading with carryover provisions, following the modeling framework of \cite{CialencoLudkovski2026}, would further enhance the practical relevance of our results. On the applied side, calibrating the model to real groundwater basins by incorporating heterogeneous crop technologies, spatially dependent trading restrictions, and empirically estimated production functions would allow regulators to quantify the trade-offs identified here in specific institutional settings. Finally, a systematic welfare comparison of the various rationing schemes, formalized through an appropriate social welfare or equity criterion, would help clarify which mechanisms are best suited to particular policy objectives.

\subsubsection*{Acknowledgments} 
IC and GDTF  acknowledge support from the National Science Foundation grant DMS-2407549. ML acknowledges support from the National Science Foundation grant DMS-2407550. 

\subsubsection*{Competing Interests} 
The authors declare no competing interests. 

\subsubsection*{Data Availability}
The developed accompanying Python code used for numerical computations is available from the corresponding author on reasonable request.


\newcommand{\etalchar}[1]{$^{#1}$}

\appendix\label{sec:apend1} 

\section{Auxiliary results } \label{sec:appendix_MM}

As mentioned in Section~\ref{sec:target_water}, the choice of scaling strength parameter $\lambda$ in \eqref{eq:MM-objective} is delicate with an interesting phase transition, see Figure~\ref{fig:price-pL-heatmap} (left).  For a fixed $A$, as the strength parameter $\lambda$ decreases to some nontrivial threshold, the optimal price $p^L$ slowly increases and then jumps to the optimal social price $p^a$; see also Figure~\ref{fig:water_sold} (right). The general message is clear: the parameter $\lambda$ must be large enough for the second term in \eqref{eq:MM-objective} to have an effect. Fundamentally, this phase transition arises from the non-convexity of the objective $\Phi^L$ in \eqref{eq:MM-objective}. As shown in Figure~\ref{fig:price-pL-heatmap} (right), as $\lambda$ decreases, the left hump decreases in value, while the right hump rises, until it eventually overtakes the left hump and becomes the global maximum.

\begin{figure}[h!]
    \centering
    \begin{subfigure}[b]{0.49\linewidth}
        \centering
        \includegraphics[width=\linewidth]{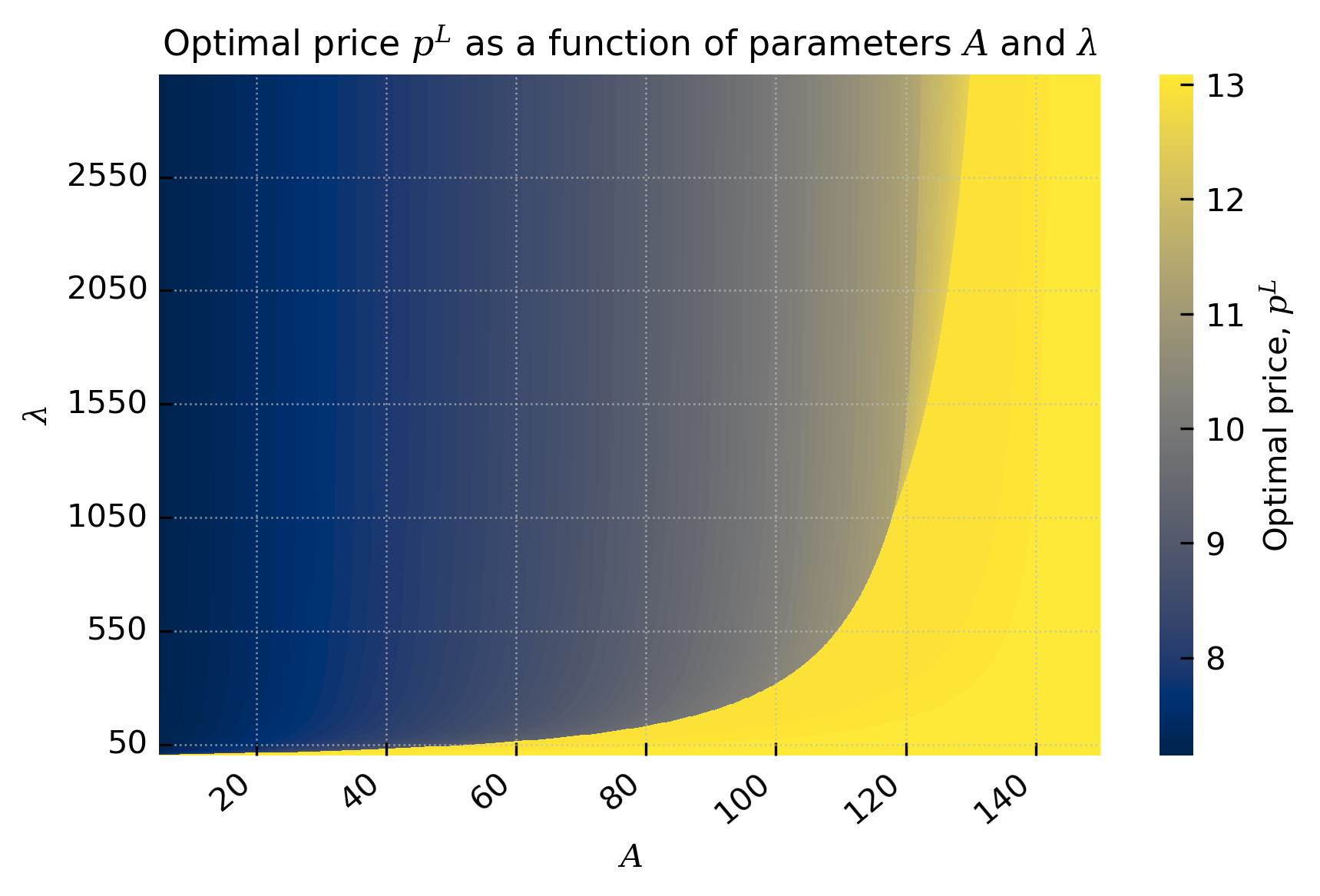}
    \end{subfigure}
    \hfill
    \begin{subfigure}[b]{0.49\linewidth}
        \centering
        \includegraphics[width=\linewidth]{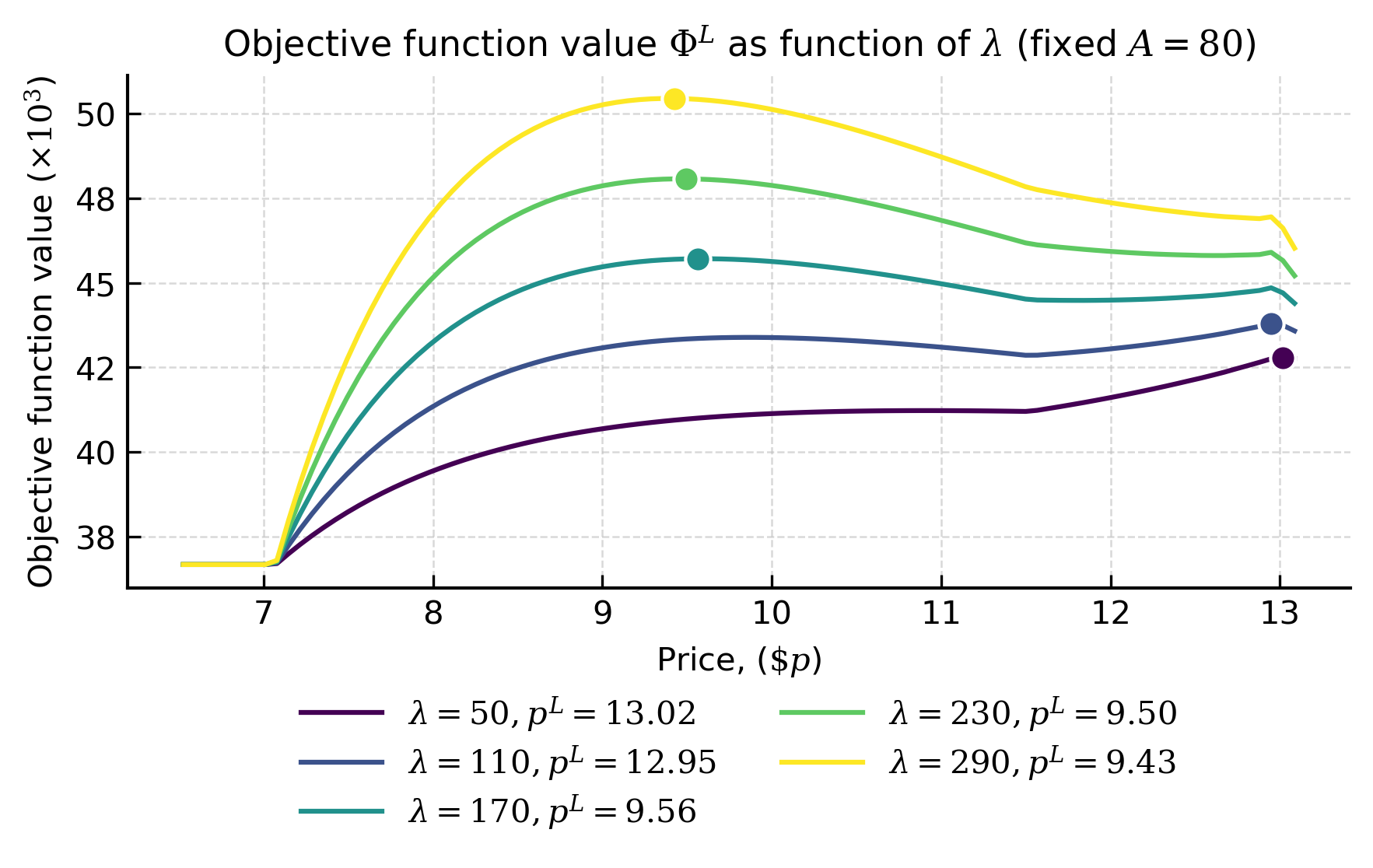}
    \end{subfigure}
\caption{Optimal price  $p^L$ as function of target  amount of water sold $A$ and strength parameter $\lambda$. }
         \label{fig:price-pL-heatmap}
\end{figure}

\section{Additional Figures}

\begin{figure}[!htbp]
    \centering
    \begin{tabular}{cc}
 \raisebox{15ex}{SF} &   \includegraphics[width=0.9\linewidth]{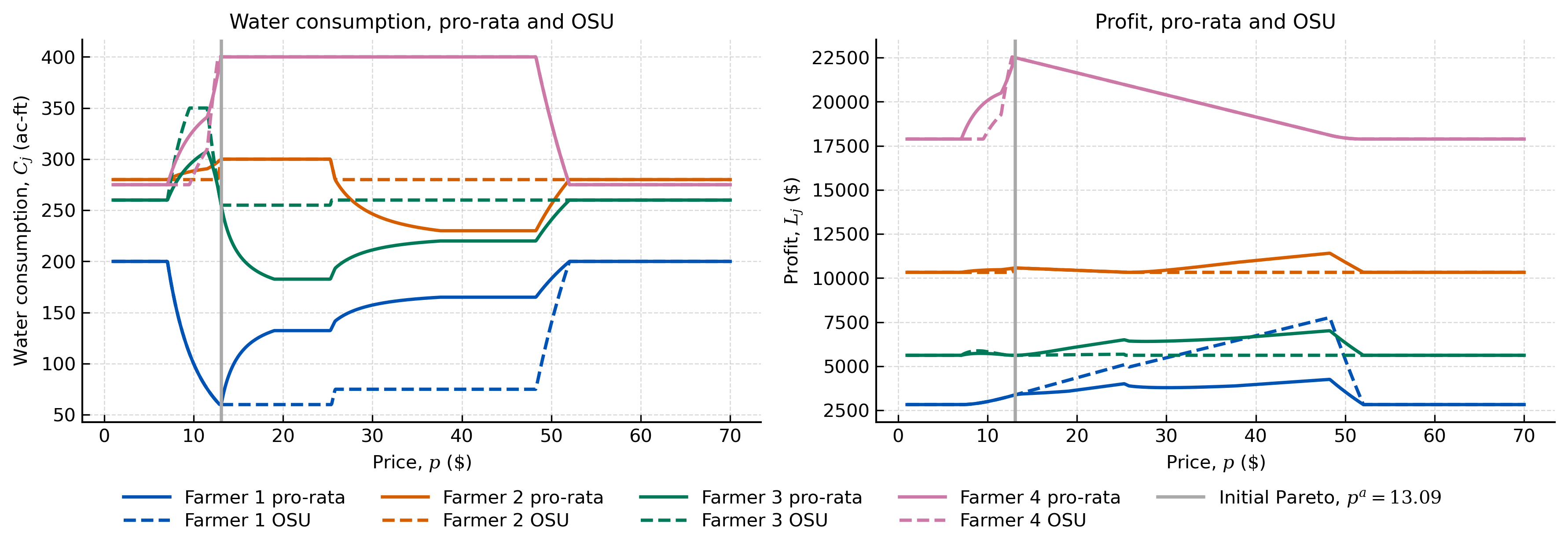} \\
 \raisebox{15ex}{JF} &   \includegraphics[width=0.9\linewidth]{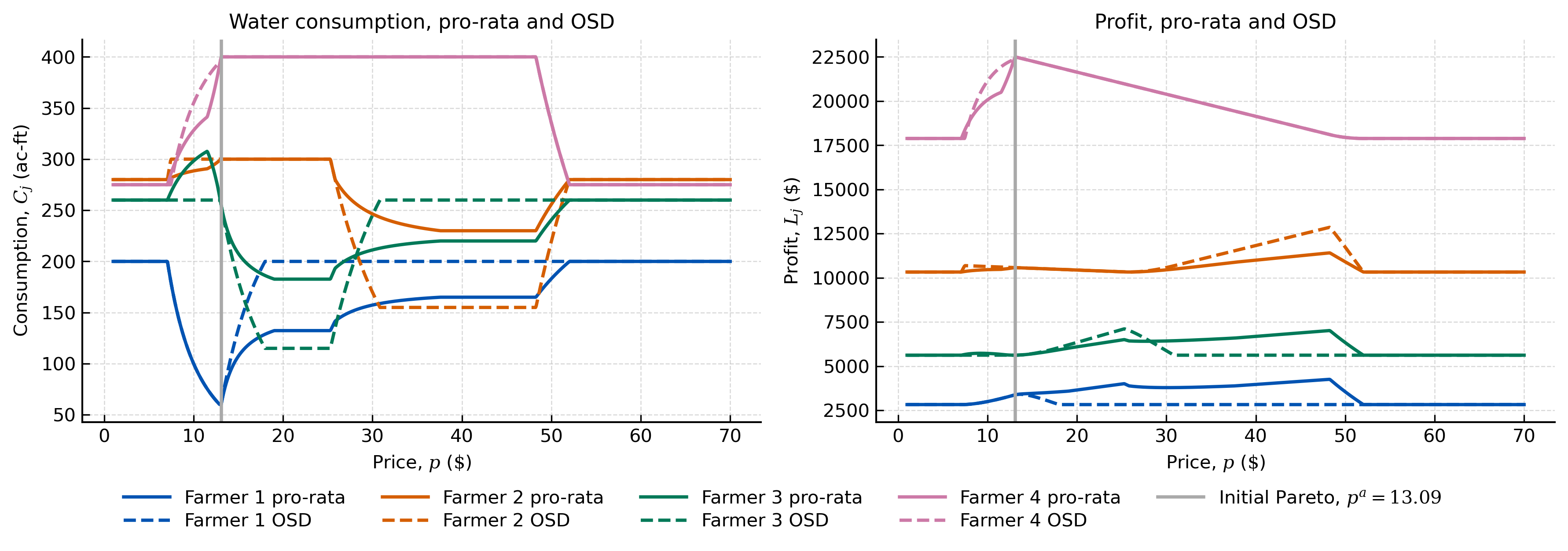} \\
  \raisebox{15ex}{UR} &       \includegraphics[width=0.9\linewidth]{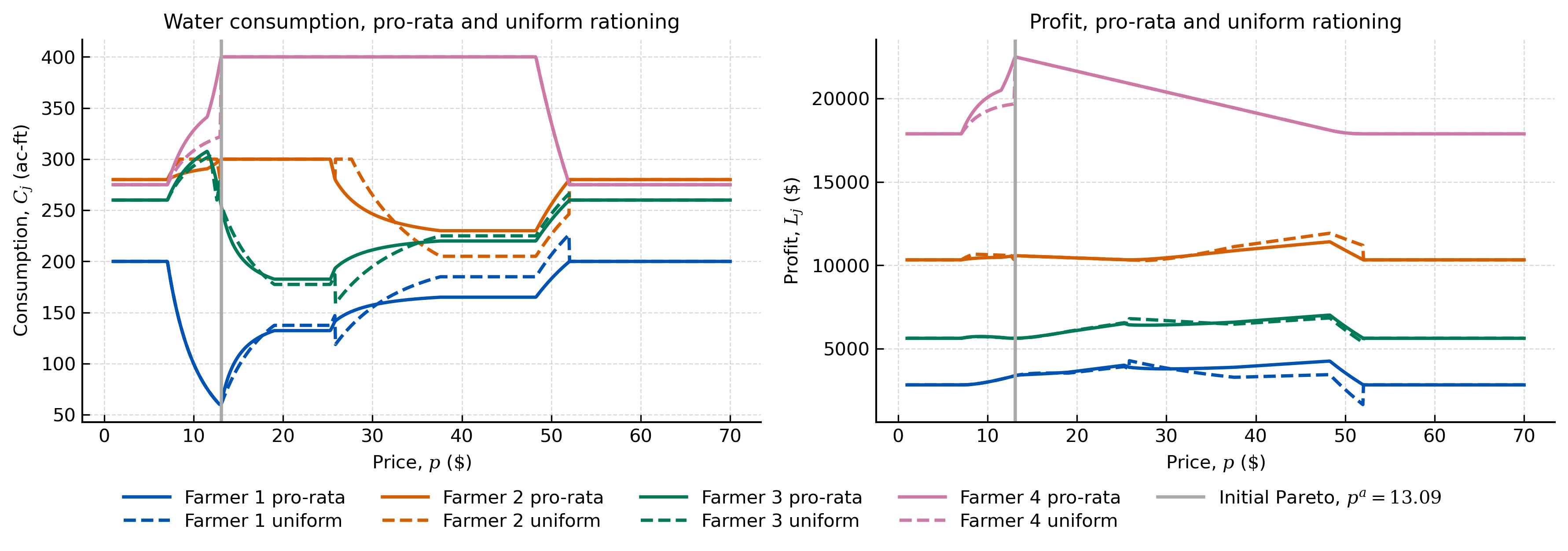}
    \end{tabular}
    \caption{Illustrating alternative allocation mechanisms: Senior-First, Junior-First, Uniform Rationing. In each row, solid lines show the base pro-rata and the dashed lines the alternative for all farmers. Left panels: consumption volumes $C_j(p)$. Right panels: profits $V_j(p)$.}
    \label{fig:OSU-all}
\end{figure}

\end{document}